\documentclass[12pt]{article}
\newtheorem{theorem}{Theorem}

\newtheorem{proposition}{Proposition}

\newtheorem{definition}{Definition}
\newtheorem{remark}{Remark}

\usepackage[toc,page]{appendix}
\newenvironment{proof}{\noindent{\bf Proof:}}{\hfill\fbox{}\vspace*{1mm}}
\usepackage[usenames]{color}
\usepackage{amsmath}
\usepackage{amssymb}
\usepackage{graphicx}
\usepackage{subfigure}
\usepackage{float}
\usepackage{booktabs}
\usepackage[table]{xcolor}
\definecolor{lightgray}{gray}{0.9}

\RequirePackage[normalem]{ulem} 
\providecommand{\DIFdeltex}[1]{{\protect\color{red}\sout{#1}}}                      
\usepackage{xcolor}
\usepackage{float}
\usepackage{tikz}
\usepackage{changebar}
\newif\ifdiff
\difftrue
\ifdiff
  \newcommand{\del}[1]{\DIFdeltex{#1}}

\else
  \newcommand{\del}[1]{}

\fi

\begin{document}
\title{\bf Generalized Optimal Liquidation  Problems Across  Multiple Trading  Venues }

\author{Qing-Qing Yang
\thanks{ Advanced Modeling and Applied Computing Laboratory,
Department of Mathematics, The University of Hong Kong, Pokfulam
Road, Hong Kong. E-mail: kerryyang920910@gmail.com.}
\and
Wai-Ki Ching
\thanks{Corresponding author. Advanced Modeling and Applied Computing Laboratory,
Department of Mathematics,
The University of Hong Kong,
Pokfulam Road, Hong Kong.
Hughes Hall, Wollaston Road, Cambridge, U.K.
School of Economics and Management,
Beijing University of Chemical Technology,
North Third Ring Road, Beijing, China.
E-mail: wching@hku.hk. }
\and Jia-Wen Gu
\thanks{Advanced Modeling and Applied Computing Laboratory,
Department of Mathematics, The University of Hong Kong, Pokfulam
Road, Hong Kong. E-mail: jwgu.hku@gmail.com.}
\and Tak-Kuen Siu
\thanks{ Department of Applied Finance and Actuarial Studies, Faculty of Business and Economics, Macquarie University, Sydney, NSW 2109, Australia.
Email: ktksiu2005@gmail.com, Ken.Siu@mq.edu.au}
}

\maketitle

\begin{abstract}
In this paper, we generalize the Almgren-Chriss's market impact model to a more realistic and flexible framework  and employ  it to derive and analyze some aspects of optimal liquidation   problem in a security market.
We illustrate how a trader's  liquidation   strategy alters  when  multiple venues and   
extra information   are brought into   the security market and detected by the trader.
This study gives some new insights into the relationship between liquidation  strategy  and market liquidity, and provides a multi-scale  approach to the optimal liquidation  problem with randomly varying volatility.
\end{abstract}

\noindent
{\bf Keywords:}
Dynamic Programming (DP);
Hamilton-Jacobi-Bellman (HJB) Equation;
Limit Order (LO);
Market Order (MO);
Multi-scale Stochastic  Volatility Model;
Quadratic Variation.


\section{Introduction}

The optimal liquidation  problem  of large trades has been studied extensively in the micro-structure literature.
Two major  sources of risk faced by large traders are:
(i)  inventory risk arising from uncertainty in   asset  value; and
(ii) transaction costs arising from market friction.
In a frictionless and competitive market, an asset can be traded with any amount at any rate without affecting the market price of the asset.
The optimal liquidation  problem then becomes an optimal stopping problem.
In an incomplete market, the optimal liquidation  problem involves delicate market-micro-structural issues.   
The impacts of transaction costs on optimal liquidation  and optimal portfolio selection have been studied via various mechanisms \cite{AC1,RA1,Almgren3, RA2, DN90, SS94}.

In this paper, we adopt a simple, but practical market impact model to study some aspects of the optimal liquidation   problem.
The  model is phenomenological and not directly based on the fine  details of micro-structure though it may primarily be related to some literature on market micro-structure.
Following Almgren and Chriss \cite{AC1}, we decompose
the price impact into temporary price impact and permanent price impact.
Temporary impact refers to temporary imbalance in supply/demand caused by   trading.
It disappears immediately when trading activities  cease.
Permanent impact means changes in the ``equilibrium'' price due to  trading,
which lasts at least for the whole process of  liquidation.

We generalize the work of  Almgren and Chriss \cite{AC1}  to a more realistic and flexible    framework  in which  multiple venues are available for the trader  to submit his/her trades.
We mainly  consider short-term liquidation problem for a large trader who experiences temporary and permanent market impact.
Two broad classes of problems are addressed in this paper which we believe are representative.
The first one is the case  with constant  volatility.
This assumption considerably simplifies the problem and allows us to exhibit the essential features of liquidation across multiple venues without losing ourselves in complexities.
The second  one is the case when volatility varies randomly throughout  the trading  horizon.
In this case,  we present a ``pure'' stochastic volatility approach,
in which the volatility is modeled as an It\^{o} process driven by a Brownian motion
that has a component independent of the Brownian motion driving the asset price.
In comparison with  Almgren's work in \cite{RA2},
we mainly focus on a special class of volatility model,
the time-scale volatility model proposed in \cite{Fouque00, Fouque11}:
$$
  d\nu_t=\epsilon(m-\nu_t)dt+\xi\sqrt{\epsilon} dW_t,
$$
and work in the regime $  \epsilon\ll 1$.
The separation of time-scales provides
an efficient way to identify and track the effect of stochastic volatility,
which is desirable from the practical perspective.

We also generalize our basic model to include the usage  of limit orders.
Different from market orders, limit orders are designed to give investors
more flexibility over the buying and selling prices of their trades.
The most unfavorable feature of limit orders might be the execution risk,
i.e., a successful execution is not guaranteed.
Cartea and Jaimungal \cite{CJ15} have constructed optimal strategies for this problem under temporary-impact-only assumption.
In their model, and also in ours, investors are allowed to trade both via limit sell orders, whose execution are uncertain, and via market sell orders,
whose execution are immediate but costly.
Limit orders in their model are allowed to be continuously submitted.
Market orders, however,  are only for one unit, and are executed at a sequence of increasing stopping times.
This model works well for small- and mid-size orders.
However, when it comes to large orders, execution cannot be guaranteed.
Furthermore, the optimal stopping time setting makes the problem difficult to be solved.
To ensure execution and  track the trace of liquidating strategy,
we model the market orders strategy as a continuous control problem.
When other market participants' buy market orders arrive,
limit sell orders will be used instead to take advantage of the price gap.

As mentioned in Rudloff et al. \cite{Rudloff},
a major reason for developing dynamic models instead of static ones is the fact that
one can incorporate the flexibility of dynamic decisions to improve our objective function.
Time-inconsistent criteria are generally not favorable to introduce in the study,
since the associated policies are sub-optimal.
The mean-variance criterion is popular for taking both return and risk into account.
However, the mean-variance criterion may induce a potential problem of
time-inconsistency, i.e., planned and implemented policies are different,
and make the problem complicated.
Hence, to take both return and risk into account, instead of adopting
the mean-variance criterion,  we are most interested in the mean-quadratic
optimal agency liquidation  strategies,
as they are proved to be time-consistent in \cite{PA1, PA2}.

In the following, we outline the major contributions of this paper.
First, we obtain closed-form solutions for optimal liquidation strategy under the constant volatility
approach for any level of risk aversion.
This leads to an efficient frontier of optimal strategies.
Each point on this frontier represents a strategy with the minimal level of cost for its level of risk-aversion.
In addition, we also  show that in the presence of multiple venues, institutional traders might hide their liquidating purposes and reduce the permanent market impacts through splitting their orders across multiple venues.
When more trading venues spring up in financial market, the liquidity of the market enhances.
This provides good suppliers of liquidity to the traders, and that the trader with more choices of venues to submit his/her trades may be willing to close his/her position as soon as possible.
Second, we present a  framework for the liquidation  problem that is general
enough to include the effect of stochastic volatility.
Moreover, it is  tractable  and the parameters
can be estimated efficiently on large datasets that are increasing available.
Third, we study the effect of incorporating limit orders so as  to detect market information.
In general, for a pure liquidation problem, people do not include the usage  of limit orders. 
But in our study of including limit orders, we find that, realizing the information carried by other market participants, traders will slow down their liquidating speed so as to get profit from market momentum.

The remainder of this paper is organized as follows.
Section 2 devotes to model building for the execution problem
in the presence of multiple trading venues.
The Profit \& Loss  (P\&L) of trading and  a mean-quadratic-variation criterion are
proposed and discussed in this section.
The solutions and numerical results for the constant volatility approximation are presented in Section 3.
Extension of the results to  the stochastic volatility approximation are then discussed in Section 4.
An extension of our model to the incorporation of limit orders is studied in Section 5.
The final section summarizes the results.

\section{Problem Setup}

In this section, we first present our liquidation problem in the presence of multiple trading venues based  on the principle of no arbitrage.
We then discuss the optimal liquidating strategies.

\subsection{The Trader's Liquidation Problem}

We consider an institutional  trader.  
Beginning at time $  t=0$, he/she has a liquidation target of $Q$ shares,
which must be completed by time $t=T$.
The number of shares remaining to liquidate at time $t$ is the trajectory $ X_t$,
with $ X_0=Q$ and $X_T=0$.
Suppose there are $N$ distinct venues for the trader to submit his/her trades.
The rate of liquidating in Venue $n$ ($n=1,2,\ldots, N$) is  $ \theta^{(n)}_t$.
Thus,  
$$
\theta^{(1)}_t+\cdots+\theta^{(N)}_t=-\frac{dX_t}{dt}.
$$
For a liquidation program, $ Q>0$, we expect $  X_t\ge 0$ is decreasing
and $  \theta^{(n)}_t\ge0$ (a buy program may be modeled similarly).
Here, $  \theta^{(n)}_t$ ($n = 1, 2, \ldots, N$), are the choice variables or
decision variables of the trader at time $t$ based on 
the observed information available at  time $t$.
In general, the trajectory $X_t$ depends on the price movements
and  market conditions that are discovered during trading, so it is a stochastic process.

Consider a probability space $ (\Omega, \mathcal F, \mathcal P)$ endowed with a filtration
$\{\mathcal F_t\}_{t\ge 0}$ representing the information structure available to the trader. Based on this probability space, we introduce the notion of admissible control.

\begin{definition}
A stochastic process  $  \boldsymbol\theta(\cdot)=\{(\theta_t^{(1)} , \ldots, \theta^{(N)} _t), 0\le t\le T\}$ is called an admissible control process if the following conditions hold:

\begin{description}
\item[(i)] {\bf(Adaptivity)} For each time  $  t\in[0,T]$, 
$\boldsymbol\theta_t=(\theta_t^{(1)},\cdots,\theta_t^{(N)}) $ is $\mathcal F_t$-adapted;
\item[(ii)]{\bf(Non-negativity)} $  \boldsymbol\theta_t\in \mathbb R_+^N$, where 
$ \mathbb R_+^N$ is the set of nonnegative real-valued  $N$-dimensional vectors;

\item[(iii)]{\bf(Consistency)} 
$$
 \displaystyle  \int_0^T\sum_{n=1}^N\theta^{(n)} _tdt\le Q;
$$
\item[(iv)]{\bf(Square-integrability)}
$$
 \displaystyle  \mathbb E\left[\int_0^T[|\theta_t^{(1)} |^2+\cdots+|\theta_t^{(N)} |^2]dt\right]<\infty;
$$

\item[(v)]{\bf($L_\infty$-integrability)}
$$
 \displaystyle  \mathbb E\left[\max_{0\le t\le T}[|\theta_t^{(1)} |+\cdots+|\theta_t^{(N)} |]\right]<\infty.
$$
\end{description}
For convenience, we let $  \Theta_t$ denote the collection of admissible controls with respect to  time $  t  (<T)$, and let  $  \widehat\Theta_t$ denote the collection of controls only  satisfying the conditions (i), (iv) and (v).
\end{definition}

Suppose the stock price evolves over time according to\footnote{When long-term strategies are considered,
it is more reasonable  to consider geometric rather than arithmetic Brownian motion.
In this paper,  we mainly  focus on short-term liquidating strategies, the total fractional price changes over such a short time are relatively  small, and hence the difference between arithmetic and geometric Brownian motions can  be negligible.}

\begin{equation}\label{basic}
  d S_t= \sigma_t dW_t
\end{equation}
where $  \{W_t\}_{t\ge0}$ is a standard Brownian motion with filtration
$\{{\cal F}_t\}_{t\ge0}$, and $  \sigma_t>0$ is the absolute volatility of the
stock price, which can be
(i) constant;
(ii)  time-dependent;
(iii) volatility depending  on the current stock price $  S_t$ and time $t$;  and
(iv) volatility  driven by an additional random process.

In Eq. (\ref{basic}), the drift term is set to be  zero, which means that  we expect no obvious trend in its future movement.
The total time $T$ is usually one day or less, the drift is generally not significant over such 
a short trading horizon.

\subsection{The Market Impact Model}

Generally, risky assets, especially for  those with low liquidity, exhibit price impacts due to the feedback effects of  trader's liquidating strategies.
Price impact refers to how the price moves in response to an order in the market. 
Small orders usually have insignificant impact, especially for liquid stocks.
Large orders, however, may have a significant impact on the price.
Investors, especially institutional investors, must keep the price impact  in mind when making investment decisions.

As discussed in Almgren \cite{RA1, RA2}, the price received on each trade is affected by the rates of buying and selling, both permanently and temporarily.
The {\it temporal impact} is related to the liquidity cost faced by the investor while the 
{\it permanent impact} is  linked to information transmitted to the market by the investor's trades.
Almgren \cite{RA2} proposed a linear market impact model to describe the dynamics of the asset price (a single venue is considered):
\begin{equation}\label{actually}
  dS^I_t= d S_t -\eta^{per}\theta_t dt.
\end{equation}
In this model, $\eta^{per}>0$ is the {\it coefficient of permanent impact} and  
is an absolute coefficient rather than fractional.
In this section, in view of  Almgren's linear price impact model,  we consider the case 
of multiple venues  based on  equilibrium and no-arbitrage.

Different from the single-venue case in Eq. (\ref{actually}),
each venue's  price in the multiple-venue case
depends not only on its internal transactions but also on its competitor's deals.
When one trader submits his orders to a venue,
he will lead to a  direct price impact in this venue.
Other traders, being aware of this impact, will adjust their order scheduling
to benefit from favorable prices across the different venues and thus affect prices.
Suppose that there are $N$ venues, where the same financial instrument can be traded simultaneously, namely Venue $1$, Venue $2$, $\ldots$, Venue $N$.
The following proposition provides a price impact model based on Almgren's linear price impact model.

\begin{proposition}
Assume that market  liquidity   remains unchanged  over $  [0,T]$,
and that there exists no arbitrage opportunity in the  financial market.
Based on the linear market impact model in Almgren \cite{RA2},
the affected  asset price  in the presence of $  N$ trading  venues is given by
$$
  dS^I_t=dS_t- \eta^{per}(\beta^{(1)} \cdot \theta^{(1)} _t+\cdots+ \beta^{(N)} \cdot \theta^{(N)} _t)dt
$$
where $  \eta^{per}>0$ is the coefficient of permanent impact,
$  (\theta_t^{(1)} , \cdots, \theta_t^{(N)} )$  is the trader's liquidating speed, and
$$
\begin{array}{l}
 \displaystyle  \sum_{n=1}^N \beta^{(n)} =1, \quad \beta^{(n)} \in(0,1)
\end{array}
$$
with  $\beta^{(n)}$ describing  Venue $n$'s  market efficiency.
\end{proposition}

\begin{proof}
Without loss of generality,  we prove that Proposition 1 holds for  $ N=2$.
Suppose that our investor's trading scheduling over a small interval
$ [t,t+\Delta t)$  is $  (\theta^{(1)} _t,\theta^{(2)}_t)$ and that
$ \theta^{(1)} _t<\theta^{(2)}_t$: liquidating $  \theta^{(1)} _t\Delta t$ shares in
Book $1$, and $ \theta^{(2)}_t\Delta t$ shares in Book $2$.
Here $  \theta^{(1)} _t$ and $  \theta^{(2)}_t$ are interpreted as liquidation rates in
Venue $1$ and Venue $2$, respectively.
Assume that trades occur immediately after time  $t$.
With the assumption of linear price impact,  stock prices in the two trading  venues,
immediately after the execution of the orders, are drawn down to
$ p_t^{(1)}=S^I_{t}-\eta^{per}\theta^{(1)} _t\Delta t$
and $p_t^{(2)}=S^I_{t}-\eta^{per}\theta^{(2)}_t\Delta t$, respectively.
Obviously, we have $  p_t^{(1)}>p_t^{(2)}$.
Other investors, being aware of this arbitrage opportunity,
will adjust their trading schedules, buying from Venue $2$  and selling to Venue $1$,
to benefit from favorable prices across the two venues.
Suppose the financial market processes linear convergence,
and the convergence speed is very quickly and proportional to market's efficiency.
The adjusted stock prices at time $  t+\Delta t$ are then given by 
$$
 \displaystyle  p_{t+\Delta t}^{(1)}
=S^I_{t}+\int_t^{t+\Delta t} \sigma_udW_u-\eta^{per}\theta^{(1)} _t\Delta t-x(\eta^{per}\theta^{(2)}_t-\eta^{per}\theta^{(1)} _t)\Delta t
$$
and
$$
  \displaystyle  p_{t+\Delta t}^{(2)}=S^I_{t}+\int_t^{t+\Delta t} \sigma_udW_u-\eta^{per}\theta^{(2)}_t\Delta t +y(\eta^{per}\theta^{(2)}_t-\eta^{per}\theta^{(1)} _t)\Delta t,
$$
respectively. 
Under the no-arbitrage principle, $p_{t+\Delta t}^{(1)}= p_{t+\Delta t}^{(2)}$
which yields $x+y=1$,
and
$$
 S^I_{t+\Delta t}
= p_{t+\Delta t}^{(1)}
=  p_{t+\Delta t}^{(2)}
= S^I_{t}+\int_t^{t+\Delta t}\sigma_udW_u-\eta^{per}(y\theta_t^{(1)} +x\theta_t^{(2)} )\Delta t.
$$
Letting  $ x=\beta$ and $ \Delta t\to 0$, we obtain
$$
  dS^I_t=dS_t-\eta^{per}[(1-\beta)\theta^{(1)}_t+\beta \theta^{(2)} _t]dt.
$$
\end{proof}

Different from the affected stock price (permanent price impact),
the actual price received on each trade (temporary price impact) varies with place and time
\begin{equation}\label{kkl}
 \widetilde S^{I,(n)}_t=S^I_t-\eta^{tem,(n)}\theta^{(n)} _t.
\end{equation}
Here $  \widetilde S_t^{I,(n)}$ is the price actually received in Venue $n$
($n=1,2,\ldots,N$),
 and $\eta^{tem, (n)}$ is the coefficient of temporary market impact in that venue.
A number of market impact models have been considered in the literature \cite{JG10},
but this simple one is good enough for our concerned problems and discussions.

\subsection{Estimation of Model Parameters}

Regarding the parameter estimation of $  \sigma_t$ in Eq. (\ref{basic}) and
$ (\eta^{per}, \eta^{tem, (n)})$  in Eq. (\ref{actually}) and Eq. (\ref{kkl}),
Almgren \cite{RA2}
proposed a variety of methods using high-frequency market data:
\begin{description}
\item[(i)] For volatility $  \sigma_t$, to filter out noise associated with market details and  obtain reliable estimates, one could estimate $ \sigma_t$  by using market data from the preceding 5 minutes, which typically would contain hundreds of trades and potentially thousands of quote updates;
\item[(ii)] One proxy for price impact parameter $ \eta^{per}$ would be the realized trade volume over the last few minutes: if more people are trading actively in the market, one would be able to liquidate a certain quantity with lower slippage;
\item[(iii)] For price impact parameters $ \eta^{tem, (n)}$, one proxy would be  the trade volume resting at or near the bid price if one is a seller (at or near the ask price if one is a buyer): a large volume there might indicate the presence of a motivated buyer and a good opportunity for one  to go in as a seller with low impact.
\end{description}

\subsection{The Gain/Loss of Trading }
 
Let $  C_t$ denote the cash flow accumulated by time $t$.
Assume that investors withhold the liquidation proceeds, or simply assume that the risk-free interest rate $  r=0$,
we have
$$
C_t
=\int_0^t(\widetilde S^{I,(1)}_s \theta^{(1)}_s+ \cdots+\widetilde S^{I,(N)}_s \theta^{(N)} _s)ds, \quad t< T.
$$
Given the state variables $  (S_t^I, C_t, X_t)$
the instant before the end of trading $  t=T-$,
we have one final liquidation (if necessary) so that the number of shares owned at
$  t=T$ is $  X_T=0$.
The liquidation value $  C_T$ after this final trade is defined to be
$$
  C_T=C_{T-}+ X_{T-}(S^I_{T-}-\mathcal C^o(X_{T-})),
$$
where $  \mathcal C^o(q)$, a non-negative increasing function in $q$,
represents the market impact  costs  the trader incurs when liquidating the outstanding position $  X_{T-}$.
The gain/loss (G/L) of trading, relative to the arrival price benchmark, is the difference between the total dollar gained/lost  by liquidating $Q$ 
shares and the initial market value: $\displaystyle   \mathcal R_T=C_T-QS_0$.
After integrating by parts and using
$$
  X_t=Q-\int_0^t(\theta^{(1)}_s + \cdots+ \theta^{(N)} _s)ds,
$$
we have
\begin{equation}\label{R}
\begin{array}{lll}
 \displaystyle \mathcal R_T&=&C_T-QS_0\\
&=&\displaystyle \int_0^{T-}\sigma_t X_tdW_t-X_{T-}\mathcal C^o(X_{T-})
\displaystyle -\eta^{per}\int_0^{T-}X_t(\beta^{(1)}\theta^{(1)}_t
+\cdots+\beta^{(N)}\theta^{(N)} _t)dt \\
&&\displaystyle - \int_0^{T-}\left[\eta^{tem,(1)}(\theta^{(1)}_t)^2+\cdots+\eta^{tem,(N)}(\theta^{(N)} _t)^2\right]dt.
\end{array}
\end{equation}

Generally speaking, investors are risk averse and demand a higher return for a more risky investment.
The mean-variance criterion is useful when taking both return and
risk into account.
However, the mean-variance criterion may induce a potential problem of
time-inconsistency, i.e., planned and implemented policies are different,
and hence complicate the problem.
To avoid  the time-inconsistent problem incurred by mean-variance criterion,
instead of using the variance/standard deviation as the risk measure,
one can adopt the quadratic variation,
$$
  \int_0^{T-}\sigma^2_uX_u^2du.
$$
It accumulates the future value of the instantaneous risk, i.e., $  (X_udS_u)^2$, 
due to holding $  X_u$ units of the risky asset $S$.
When properly normalized, the quadratic variation can also be interpreted as
the average standard deviation per unit time, see, for instance, Brugiere \cite{PB96}.

Let $  (S_t, \sigma_t, X_t)=(s, \sigma, q)$ be the initial state.  At any time $  t\in[0, T)$,  the optimal policy should solve the following optimization problem:
\begin{equation}\label{mean-quadratic}
 J(t, s,\sigma,  q)=\max_{(\theta^{(1)}_u,\cdots, \theta^{(N)} _u)_{ t\le u\le T}\in \Theta_t}\left\{ \mathbb E_{t}[\mathcal R_T-\mathcal R_t] -\lambda\cdot\mathbb E_t\left[\int_t^{T-} \sigma^2_uX_u^2du\right]\right\},
\end{equation}
where $  \mathbb E_t[\cdot]$ denotes the conditional expectation with respect to the filtration $\mathcal F_t$. 
The following proposition discusses the time consistency of the optimal strategies
and its proof is given in Appendix A.

\begin{proposition} {\bf (Time consistency of the optimal strategies).}
 Let $  (t_1,s_1,\sigma_1,  q_1)$ be some state at time $  t_1$ and
$ \boldsymbol\theta^*_{t_1,s_1, \sigma_1, q_1}(\cdot)
=\big(\theta^{(1)}_{t_1,s_1,\sigma_1, q_1}(\cdot), \ldots, \theta^{(N)} _{t_1,s_1, \sigma_1,  q_1}(\cdot)\big)^T$
be the corresponding optimal strategy.
Let  $  (t_2, s_2, \sigma_2, q_2)$ be some other state at time $  t_2>t_1$ and
$
  \boldsymbol\theta^*_{t_2, s_2, \sigma_2, q_2}(\cdot)=\big(\theta^{(1)}_{t_2, s_2, \sigma_2, q_2}(\cdot), \ldots, \theta^{(N)} _{t_2, s_2, \sigma_2, q_2}(\cdot)\big)^T
$
be the corresponding optimal strategy.
It follows that the optimal controls of Problem (\ref{mean-quadratic}) are {\it time-consistent} in the sense that for the same state 
$(t^\prime, s^\prime, \sigma^\prime, q^\prime)$ at a later time 
$t^\prime>t_2$,
 \begin{equation}\label{proof}
   \boldsymbol\theta^*_{t_1,s_1, \sigma_1, q_1}(t^\prime, s^\prime, \sigma^\prime,  q^\prime)= \boldsymbol\theta^*_{t_2, s_2, \sigma_2,  t_2}(t^\prime, s^\prime, \sigma^\prime, q^\prime).
  \end{equation}
\end{proposition}

\subsection{Hamilton-Jacobi-Bellman (HJB) Equation}

Since the optimal controls satisfy the Bellman's principle of optimality as shown in Appendix A, dynamic programming (DP) approach can be directly applied to this problem.
The optimal control $\boldsymbol\theta^*$ can be obtained by solving the following
 Hamilton-Jacobi-Bellman (HJB)  equation derived in Appendix A as follows:
$$
\mbox{(HJB-1)}\quad\left\{
\begin{array}{lll}
  \displaystyle  (\partial_t +\mathcal L) J    -\lambda  \sigma_t^2q^2\\
\displaystyle
  -   \min_{\boldsymbol\theta_t\in\Theta_t } \Big\{  \eta^{per}(\beta^{(1)}\theta^{(1)}_t+\cdots+\beta^{(N)}\theta^{(N)} _t)q+ \eta^{tem,(1)}(\theta^{(1)}_t)^2+\cdots+\eta^{tem,(N)}(\theta^{(N)} _t)^2\\
   \displaystyle +(\theta_t^{(1)} +\cdots+\theta^{(N)} _t)\; \partial_q J\ \Big\} =0\\
  J(T-,s, \sigma,  q)= -q\mathcal C^o(q),
\end{array}
\right.
$$
where $\mathcal L$ is the generator of the processes $  (S_t, \sigma_t)_{t\ge0}$.

Notice that the optimization problem included in  HJB-1 is a constrained  optimization problem with  constraints: (1) $\boldsymbol\theta_t\in\mathbb R_+^N$ and
(2) $\displaystyle \int_0^T\sum_{n=1}^N\theta_t^{(n)} dt\le Q$.
To solve this constrained optimization problem, we first consider relaxing  
the  constraints, namely, replacing $\Theta_t$ by $\widehat\Theta_t$,  and solving the following unconstrained  optimization problem:
$$
\mbox{(HJB-$1^\prime$)}\quad \left\{
\begin{array}{lll}
  \displaystyle  (\partial_t +\mathcal L) J    -\lambda  \sigma_t^2q^2\\
\displaystyle
   -   \min_{\boldsymbol\theta_t\in\widehat\Theta_t } 
\Big\{  \eta^{per}(\beta^{(1)}\theta^{(1)}_t+\cdots+\beta^{(N)}\theta^{(N)} _t)q+ \eta^{tem,(1)}(\theta^{(1)}_t)^2+\cdots+\eta^{tem,(N)}(\theta^{(N)} _t)^2\\
  \displaystyle +(\theta_t^{(1)} +\cdots+\theta^{(N)} _t)\; \partial_q J\ \Big\} =0\\
   J(T-,s, \sigma,  q)= -q\mathcal C^o(q),
\end{array}
\right.
$$
 and then prove that under some assumptions, the two optimization problems, 
included in HJB-1 and HJB-$1^\prime$, respectively,  are equivalent.
From the HJB-$1^\prime$ equation,
the optimal control without any constraint  is
\begin{equation}\label{ef}
 \theta_t^{n,*}=-\frac{1}{2\eta^{tem, (n)}}(\partial_q J+\eta^{per}\beta^{(n)} q), 
\quad {\rm for} \ n=1,\cdots, N.
\end{equation}
The corresponding  value function $J$ then solves the following 
partial differential equation (PDE):
\begin{equation}\label{PDE}
\left\{
\begin{array}{l}
    \displaystyle  (\partial_t +\mathcal L) J    -\lambda  \sigma_t^2q^2+\sum_{n=1}^N \frac{1}{4\eta^{tem,(n)}}(\partial_q J+\eta^{per}\beta^{(n)} q)^2=0\\
    \displaystyle J(T-,s,\sigma, q)=-q \mathcal C^o(q).
  \end{array}
  \right.
\end{equation}
\quad\\

In the following  sections, we  exhibit  solutions to Eq. (\ref{PDE})  with linear penalty: 
$\mathcal C^o(q)=Kq$  to two special cases:
(i) constant volatility and (ii) {\it slow mean-reverting} stochastic volatility.
In each case, we work on the ansatz that
$$
  J(t, s,\sigma, q)=f(t, s,\sigma)+g(t,s,\sigma) q+h(t, s, \sigma) q^2.
$$
The first case has explicit liquidating formula.
For the {\it slow mean-reverting} stochastic volatility approximation discussed in this paper,  there is a multi-scale argument that reduces the optimal
liquidation PDE to a formal series expansion that can easily be solved explicitly.
We present the argument and the accuracy of this approach in Section 3.3.

\section{Constant Volatility}

The most illuminating case is that $  \sigma_t$ is constant, i.e., $  \sigma_t\equiv \sigma$.
This considerably simplifies the problem and allows us to exhibit the essential features of liquidating across multiple venues without losing ourselves in complexities.
The problem becomes essentially the well-known stochastic linear regulator with time dependence.

With the constant volatility assumption, $J(t, s, \sigma, q)=J(t,  q; \sigma )$ independent of $s$, since  the terminal data does not depend on $s$ and  
HJB-$1^\prime$    introduces no $s$-dependence.
We look for a solution quadratic in the inventory variable $q$:
\begin{equation}\label{lb1}
  J(t,q;\sigma)=f(t;\sigma)+g(t;\sigma)q+h(t;\sigma)q^2.
\end{equation}
With this assumption, the optimal control can be rewritten in the form of 
\begin{equation}\label{rew}
  \theta_t^{n,*}=-\frac{1}{2\eta^{tem, (n)}}[(2h(t;\sigma)+\eta^{per}\beta^{(n)})q
+g(t;\sigma)].\end{equation}
To solve Eq. (\ref{PDE}), $f, g$ and $h$ must satisfy the following  ordinary differential equations (ODEs):
\begin{equation}\label{ODE}
 \begin{array}{l}
 \left\{
 \begin{array}{l}
   \displaystyle \dot{h} =\lambda\sigma^2-\sum_{n=1}^N\frac{1}{4\eta^{tem, (n)}}(2h
+\eta^{per}\beta^{(n)})^2\\
  \displaystyle h(T-;\sigma)=-K
 \end{array}
 \right.\\
 \\
 \left\{
 \begin{array}{l}
  \displaystyle \dot{g}=-\left(\sum_{n=1}^N\frac{1}{2\eta^{tem, (n)}}(2h+\eta^{per}\beta^{(n)})\right)g\\
  \displaystyle g(T-;\sigma)=0
 \end{array}
 \right.\\
 \\
 \left\{
 \begin{array}{l}
  \displaystyle \dot{f}=-\left(\sum_{n=1}^N\frac{1}{4\eta^{tem, (n)}}\right)g^2\\
  \displaystyle f(T-;\sigma)=0.
 \end{array}
 \right.
 \end{array}
 \end{equation}
 It is straightforward to show that  $  g(t;\sigma)\equiv 0$ and $  f(t;\sigma)\equiv 0$.
If we  set
 $$
  a_n=\frac{1}{\eta^{tem, (n)}},  \quad
b_n=\frac{\eta^{per}\beta^n}{\eta^{tem,(n)}}, \quad
c_n= \frac{(\eta^{per}\beta^n)^2}{ 4\eta^{tem,(n)}},
$$
and 
$$
  a=\sum_{n=1}^Na_n,\quad
  b=\sum_{n=1}^Nb_n,\quad  c=\sum_{n=1}^N c_n,
$$
then the unknown function $ h(t;\sigma)$ solves the following first-order  ODE:
\begin{equation}\label{ho}
\left\{
\begin{array}{l}
  \displaystyle \dot{h}=\underbrace{\lambda\sigma^2+\frac{b^2-4ac}{4a}}_{\Delta_N}-a\left(h+\frac{b}{2a}\right)^2\\
  \displaystyle h(T-;\sigma)=-K.
\end{array}
\right.
\end{equation}
Eq. (\ref{ho}) is a first-order ODE with  constant coefficients, and  we can find the exact solution via direct integrations: 
\begin{equation}\label{h}
  h(t;\sigma)=
\left\{
\begin{array}{lll}
  \displaystyle \sqrt{\frac{\Delta_N}{a}}\cdot \frac{\varsigma e^{-2\sqrt{a\Delta_N}(T-t)}-1}{\varsigma e^{-2\sqrt{a\Delta_N}(T-t)}+1}-\frac{b}{2a},&&  \Delta_N\ge0\\
\\
  \displaystyle \sqrt{-\frac{\Delta_N}{a}}  \tan\Big(\arctan\big(\frac{b-2a K}{2\sqrt{-a\Delta_N}}\big)+ \sqrt{-a\Delta_N} (T-t)\Big) -\frac{b}{2a},&&  \Delta_N<0,
\end{array}
\right.
\end{equation}
where   the constant  $\varsigma$  is given by
$$
  \varsigma= \left(1-\frac{2Ka-b}{2\sqrt{a\Delta_N}}\right)\left(1+\frac{2Ka-b}{2\sqrt{a\Delta_N}}\right)^{-1}.
$$

 It is worth noting that 
$$
 \theta_t^{n,*}=-\frac{1}{2\eta^{tem, (n)}}(2h(t;\sigma)+\eta^{per}\beta^{(n)})X_t.
$$
and 
$$
\left\{
\begin{array}{l}
  \displaystyle \dot{X}_t=-\sum_{n=1}^N\theta^{n,*}_t\\
  \displaystyle X_0=Q.
\end{array}
\right. 
$$
Hence, we have
$$
 X_t
= Q\cdot \exp\left(\displaystyle \int_0^t\sum_{n=1}^N\frac{1}{2\eta^{tem, (n)}}(2h(u;\sigma) +\eta^{per}\beta^{(n)}) du\right).
$$
For $n=1,2,\ldots, N$, we have $b_n^2=4a_nc_n$.
According to H$\ddot{\rm o}$lder's inequality,   we have
$b^2-4ac\le 0$, and  the equality holds if and only if
(1) $  N=1$; or (2) the trading  venues have the same market efficiency, namely, 
$ \beta^{(1)}=\cdots=\beta^{(N)}=1/N$. 
We then have the following proposition and its proof is given in Appendix B.
\begin{proposition}
Assume that   the model parameters satisfy the condition\footnote{That is,  clearing fees associated with  the outstanding position $  X_{T-}$ dominate the  potential profit arising from arbitrage opportunities incurred by the permanent price impact and the  potential position risk involved by price fluctuations.}:
\begin{equation}\label{cond1}
  K>\displaystyle  \frac{b}{2a}+\sqrt{\frac{|\Delta_N|}{a}}.
\end{equation}
If $N=1$ or  the trading venues have the same market  efficiency, 
then $ h(t;\sigma)$ is a decreasing  function in $t$ and  
$2h(t;\sigma)+\eta^{per}\beta^{(n)}\le0$ for $n=1,\cdots, N$, which implies 
that 
\begin{description}
\item[(1)] $\theta^{n,*}_t \ge0$,  for any $  n=1,\cdots, N$; and that 
\item[(2)] $\displaystyle \int_0^T\sum_{n=1}^N \theta^{n,*}_tdt \le Q$.
\end{description}
Then the    control policy in Eq. (\ref{rew}) is optimal. 
\end{proposition}
Since
 $$
 \frac{\partial h(t;\sigma)}{\partial t} =-\frac{\partial h(t;\sigma)}{\partial \tau}<0,
 $$
 where $  \tau=T-t$ is the remaining  time to close of trading. 
  $  J(t,q; \sigma)=h(t;\sigma)q^2$ is a strictly decreasing function in $t$ and a strictly increasing function in $T$. 
  
Generally speaking, one's ability to bear risk is measured mainly in terms of objective factors, such as time horizon, risk aversion and expected income. 
Let $  J_T(t,q;\sigma)$ denote the value function of the optimization problem 
(\ref{mean-quadratic}) with time horizon $T$,  for any $  T_1>T_2>t$, we have
$J_{T_1}(t,q;\sigma)> J_{T_2}(t,q;\sigma)$.
This coincides with the actual situation that an investor's risk affordability directly relates to his/her time horizon. 
An investor with a two-week time horizon can be considered to have a greater ability to bear risk, other things being equal, than an investor with a two-hour horizon. 
This difference is because over two weeks there is more scope for losses to be recovered or other adjustments to circumstances to be made than there is over two hours. 

 \subsection{Effect of Multiple Venues}

 With the development of electronic exchanges, many new trading destinations have appeared to compete the trading capability of the fundamental financial markets such as the NASDAQ's Inet and NYSE in the US, or EURONEXT, the London Stock Exchange and Xetra in Europe. 
As a result, the same financial instrument can be traded simultaneously in different venues. 
These trading venues are generally different from each other at any time because of variation in the fees or rebates they demand to trade and the liquidity they offer. Therefore, to liquidate a large order, traders may need to split their order across various venues to reduce market impact. 
In this subsection, we illustrate theoretically  how multiple venues affect an investor's trading strategy.

Denote the liquidating strategy  corresponding to a single venue case as
$  \theta_t^{single,*}$.
Suppose  that there are $N$ distinct venues for an investor to submit his/her trades and that   condition (\ref{cond1}) is satisfied. Without loss of generality, we assume   that there is no difference among there  trading venues, namely $  \beta^1=\cdots=\beta^N=1/N$ and $  \eta^{tem, (1)}=\cdots=\eta^{tem, (N)}=\eta^{tem}$.
 The following two conclusions about the effects on optimal liquidating speed 
and transaction cost can be drawn.

\begin{description}
\item[(i)]  {\bf(Effect on optimal liquidating speed).} First, we observe that
$ \theta^{1,*}_t=\cdots=\theta^{N,*}_t:=\theta^*_t(N)$, where
$$
 \theta^*_t(N)= - \sqrt{\frac{\lambda\sigma^2 }{\eta^{tem}N}}\times \frac{\varsigma(N)e^{-2\sqrt{\lambda\sigma^2N/\eta^{tem}}(T-t)}-1}{\varsigma(N)e^{-2\sqrt{\lambda\sigma^2N/\eta^{tem}}(T-t)}+1}X_t:=\mathfrak{J}(t, N)X_t,
$$
with 
$$
\varsigma(N)=\frac{2\sqrt{\frac{\lambda\sigma^2N}{\eta^{tem}}}-(2K\frac{N}{\eta^{tem}}-b)}{2\sqrt{\frac{\lambda\sigma^2N}{\eta^{tem}}}+(2K\frac{N}{\eta^{tem}}-b)}.
$$
Notice that $  X_0=Q$ and  that 
$$
  \lim_{N\to\infty} \dot{X}_0=-\lim_{N\to\infty} N\theta^{*}_0(N)=-\lim_{N\to\infty} N\mathfrak{J}(0,N)Q=\infty.
$$ 
For any fixed time  $  t\in(0,T)$, we have 
$$
 \lim_{N\to\infty}\mathfrak{J}(t, N)=0
\quad {\rm and} \quad 
  \lim_{N\to\infty}\sqrt{N}\mathfrak{J}(t,N)=\sqrt{\frac{\lambda\sigma^2}{\eta^{tem}}}.
$$ 
Hence
$$
\begin{array}{lll}
  \displaystyle \lim_{N\to\infty} X_t&=&  \displaystyle\lim_{N\to\infty}  Q e^{-\int_0^tN\mathfrak{J}(u,N)du} 
= \displaystyle \lim_{N\to\infty} Q e^{-\sqrt{N}\times \int_0^t\sqrt{N} \mathfrak{J}(u,N)du}=   \displaystyle \lim_{N\to\infty} Q e^{-a_1 \sqrt{N}} =0,
\end{array}
$$
where\footnote{According to the structure of $\mathfrak{J}(t, N)$,  there exists a finite number $M>0$ such that
$$
|\sqrt{N}\mathfrak{J}(t, N)|\le M,\quad \mbox{for all $(t, N)\in[0,T)\times \mathbb Z_+$,}
$$
where $\mathbb Z_+$ is the set of all nonnegative integers. 
Directly applying the Dominated convergence theorem to $\{\sqrt{N}\mathfrak{J}(t, N)\}$, we  obtain the result in Eq. (\ref{pk}).   
}
\begin{equation}\label{pk}
  0<a_1=\lim_{N\to\infty} \int_0^t\sqrt{N}\mathfrak{J}(u, N)du=\sqrt{\frac{\lambda\sigma^2}{\eta^{tem}}} t<\infty.
\end{equation}
We also have
$$
\begin{array}{lll}
 \displaystyle \lim_{N\to\infty} -\dot{X}_t &=&  \displaystyle \lim_{N\to\infty}N\theta^{*}_t(N)\\
&=&  \displaystyle \lim_{N\to\infty}QN\mathfrak{J}(t,N)\times e^{-\int_0^tN\mathfrak{J}(u,N) du }\\
&=&  \displaystyle \lim_{N\to\infty} Q \frac{\sqrt{N}\times\sqrt{N}\mathfrak{J}(t,N)}{e^{\sqrt{N}\times\int_0^t\sqrt{N}\mathfrak{J}(u, N)du}}\\
&=&   \displaystyle \lim_{N\to\infty} Q\sqrt{\frac{\lambda\sigma^2}{\eta^{tem}}}\times   \frac{\sqrt{N}}{e^{a_1\sqrt{N}}}\\
&=& \displaystyle \lim_{N\to\infty} Q\sqrt{\frac{\lambda\sigma^2}{\eta^{tem}}}\times   \frac{1}{a_1e^{a_1\sqrt{N}}}\\
&=&0.
\end{array}
$$
That is to say,  as $N$ approaches infinity, the investor would immediately close his/her position at the beginning of the trading horizon.

\item[(ii)]{\bf(Effect on transaction cost).}
The two strategies: 1) $  \{\theta_t^{single,*}\}_{t\in[0,T]}$, liquidating in a single venue;  and 2)
$$
  \Big\{\theta^{(1)}_t=\theta^{(2)}_t=\cdots=\theta^{(n)} _t= \frac{1}{N}\theta_t^{single,*}\Big\}_{t\in[0,T]},
$$
equally splitting the original target among $N$ venues, 
transmit the same information to the market (i.e., have the same permanent impact),
but involve different transaction costs (i.e., have different temporary price impacts)
$$
\begin{array}{lll}
  \displaystyle\underbrace{   \int_0^{T-}\sum_{n=1}^N\left(S_t^I-\widetilde S_t^{I,(n)}\right)\theta^{(n)} _tdt}_{\mbox{costs in $N$-venue}} &=&  \displaystyle\int_0^{T-}
\left(\sum_{n=1}^N\frac{\eta^{tem}}{N^2}\right)(\theta_t^{single,*})^2dt\\
&=&  \displaystyle\int_0^{T-} \frac{\eta^{tem}(\theta_t^{single,*})^2}{N}dt\\
&\le&  \displaystyle \int_0^{T-}\eta^{tem}(\theta_t^{single,*})^2dt=\underbrace{  \int_0^{T-}\left(S_t^I-\widetilde S_t^I\right)\theta_t^{single,*}dt.}_{\mbox{costs in a single venue}}\\
\end{array}
$$
That is, liquidating schedule across multiple venues can indeed help to reduce  transaction costs arising from market liquidity.
\end{description}

\subsection{Numerical Results}

In this section, we provide some numerical  results to illustrate the effects of different market factors on investor's liquidation strategy.
We assume that there are $N$ distinct trading venues for an investor to submit his/her trades, and that there is no difference among these venues.
As far as our simulation is concerned, we choose the  following hypothetical values for the model parameters:
$$
   \sigma=\exp(1), \ \ T=1, \ \ K=0.1,  \ \ \eta^{per}=0.005, \ \ \eta^{tem}=0.01.
$$
The risk-aversion parameter $\lambda$ ranges across all nonnegative values,
as the actual choice of trajectory will be determined by the trader's risk preference.

\subsubsection{Effect of Multiple Venues}

When multiple trading venues are available, dividing a target quantity across these venues may help an investor to hide his liquidation purpose and hence reduce the permanent market impact.  
To facilitate our analysis, we assume that  these trading venues are identical for the investor.
Some numerical results are presented in Table 1.

\begin{table}[H]
\caption{1, 000 Simulations  with $  Q=100, s=15, \lambda=0.1$.}
\begin{center}
\begin{tabular}{cccccc}
\hline
\cline{1-6}
& &Std&&Std&   \\
\#\{Venues\} &G/L  & (G/L$)^*$&$  q_T$&$  (q_T)^{**}$&Value function\\
\hline
 1&-461.80&63.38&1.43&2.22& -902.90\\
 2&-324.86&54.19&0.03&0.04&-638.30\\
 3&-265.08&46.12&0 &0 &-522.83\\
4& -230.51&  43.03&  0 &  0 & -455.06\\
 10& -147.22& 33.51&0&0   &-293.05\\
 \vdots&\\
 50&-72.92&18.57&0&0&-144.59\\
\hline
\cline{1-6}
{\scriptsize*: in $10^{-2}$}\\
{\scriptsize**: in $10^{-16}$}\\
\end{tabular}
\end{center}
\label{table1}
\end{table}

Table 1 provides a comparison of optimal  liquidation  strategy among  investors facing
multiple venues.
The first column shows the number of venues available for an investor to submit 
his/her trades.
The second and third columns show, respectively,  the mean and standard derivation  of 
G/L (Eq. (\ref{GL})). 
The negative G/L implies  that the investor is trading    at a loss. The absolute value represents the {\it cost of trading}. As defined so far, the {\it cost of trading}, relative to the arrival price benchmark, is the difference between the total dollars received to liquidate Q shares and the initial market value:
\begin{equation}\label{GL}
  {  G/L}= \int_0^T(\widetilde S_t^{I,(1)}\theta^{(1)}_t
+\cdots+\widetilde S_t^{I,(N)}\theta^{(N)} _t)dt+X_{T-}(S_{T-}^I-\mathcal C^o(X_{T-}))-QS_0\quad(<0).
\end{equation}
It  is an indicator of  market liquidity.   
The greater the loss, the worse the situation. 
The last column shows the corresponding value taken by the  value function:
$$
  {  G/L}-\lambda  \mathbb E\left[\int_0^T \sigma_t^sX_t^2dt \right] \quad(<0),
$$
an indicator of the investor's satisfaction. 
Investors are risk-averse ($\lambda>0$), so usually they do not hold outstanding stock 
uncleared at the end of trading $  t=T-$. 
The fourth and fifth columns show, respectively, the mean and standard deviation of the outstanding position at the end of trading.
The effect of  multiple venues on investors' liquidation strategies  is fairly straightforward.
With the increase of  trading  venue,  investors'  loss of  trading  
 decreases. 
Meanwhile, as the  number of  trading  venue increases, the risks involved decreases, which in turn enhances investors' satisfaction (measured by the value function).

\subsubsection{Trading Curve}

The average number of shares at each point of time, say the trading curves,  with
respect to different number of multiple venues are depicted  in Figure \ref{fig2}.
\begin{figure}[H]
\begin{center}
\includegraphics[width=5.in, height=4.5in]{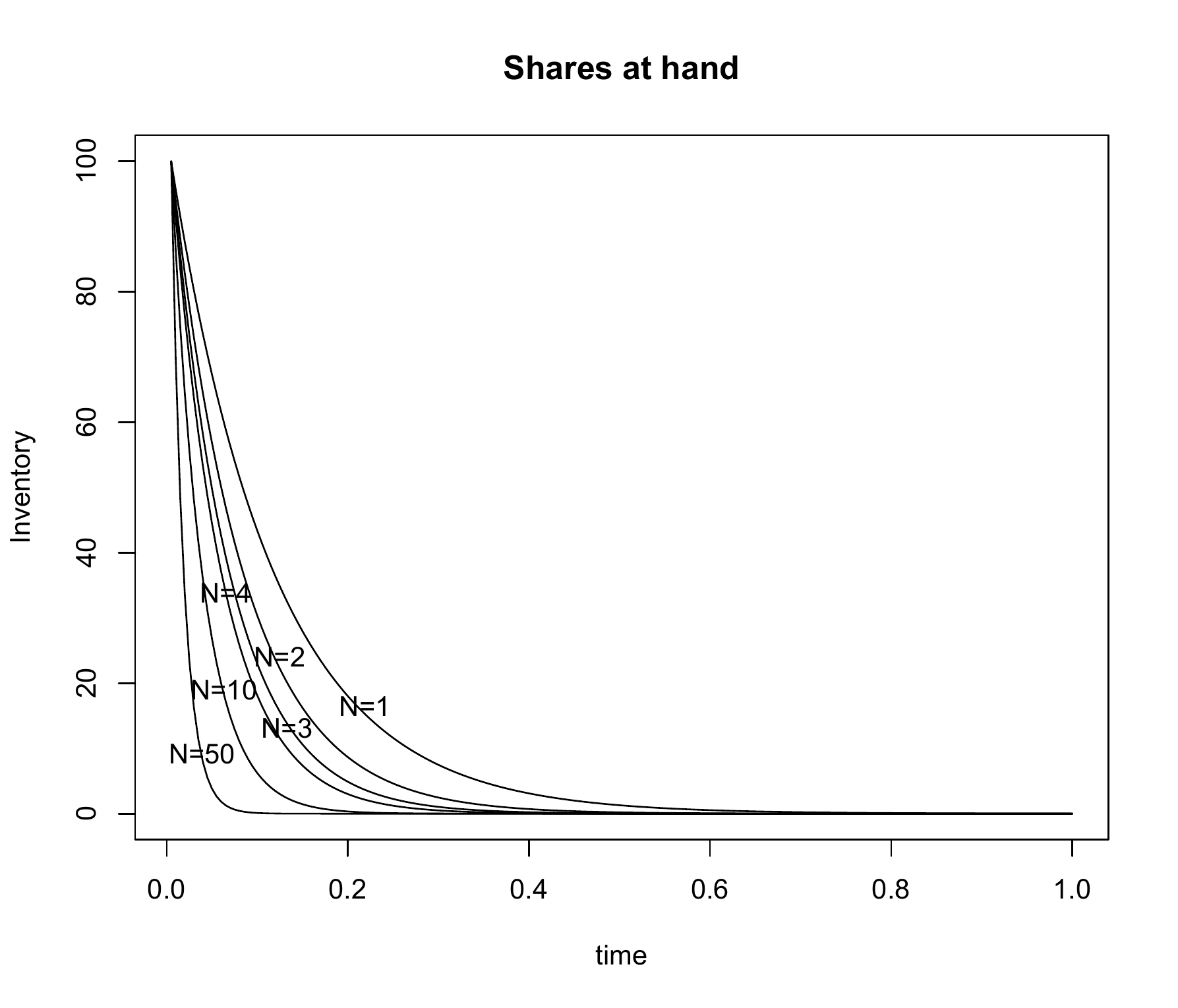}
\caption{  Trading curves with $  Q=100$.  }
\label{fig2}
\end{center}
\end{figure}
It is clear that as the number of alternative choices increases,  the average liquidation speed increases.
Indeed,  a trader needs to trade fast to reduce position risk,
when the quantity to liquidate is large within  the remaining time.
However,  an immediate execution is often not possible or at a very high cost due to insufficient liquidity.   The trading cost  reflects the difference between the amount at which the trader expects to sell and the  sales proceeds the trader actually receives.
It is an indicator of the illiquidity of a market.
We can see from the results  that, when    more trading venues spring up in the financial market,
the liquidity of the asset  enhances.
This provides good supplies of liquidity to the trader, and that  the trader with more choices of venues to submit his/her trades may be willing to close his/her position earlier.
These observations are consistent with financial intuition.

%
%
%
%
%
%

\subsubsection{Efficient Frontier}

The efficient frontier consists of all optimal trading strategies. Here the ``optimal'' refers to the situation where no strategy has a smaller variance for the same or higher level of expected transaction profits. For each level of risk-aversion,  $\lambda$,  there is an optimal liquidating strategy.
By running 1,000 simulations with initial inventory $  Q=100$, we obtain an efficient frontier
(see  Figure \ref{fig3}).
\begin{figure}[H]
\begin{center}
\includegraphics[width=5.in, height=4.5in]{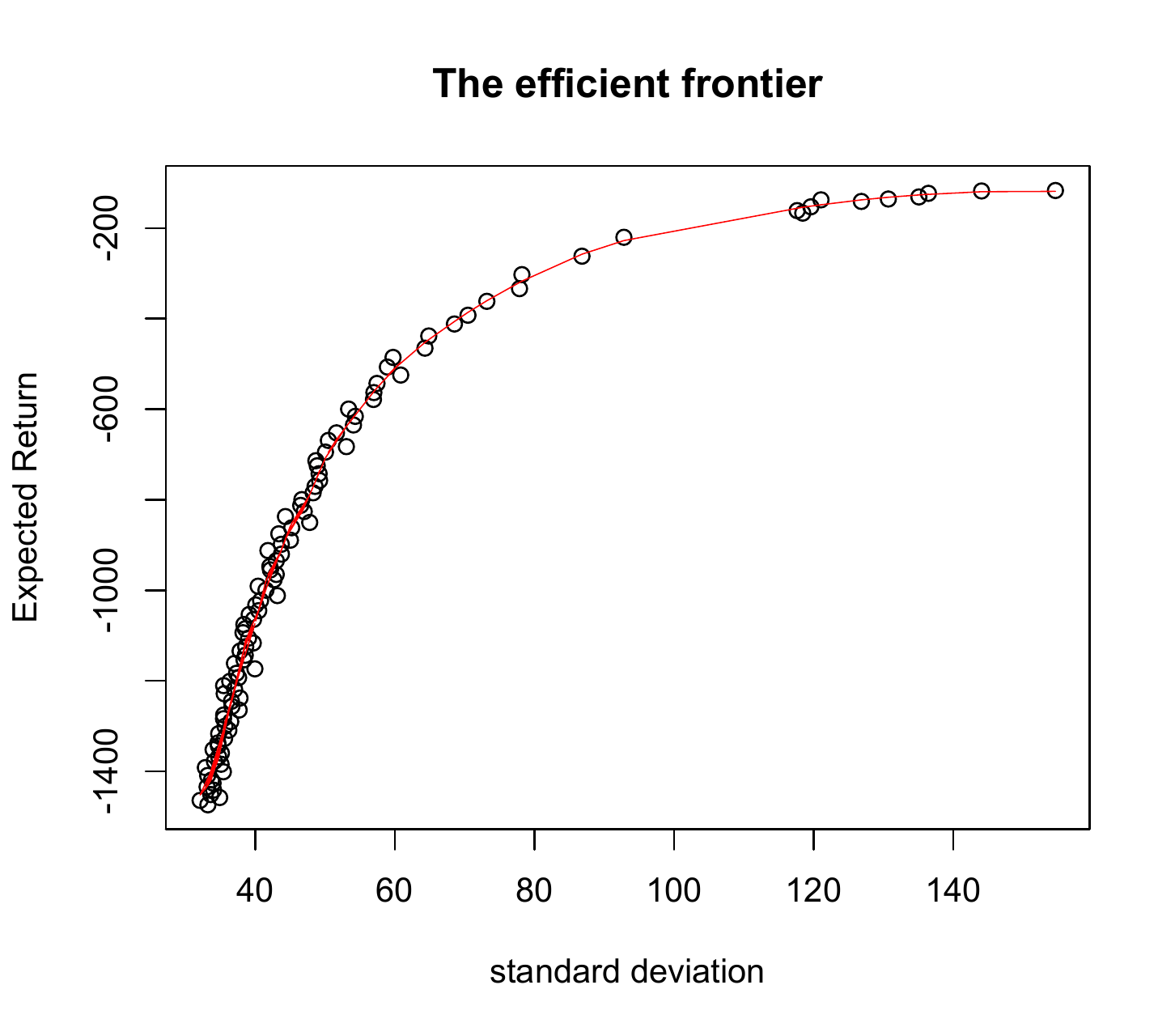}
\caption{The efficient frontier.  }
\label{fig3}
\end{center}
\end{figure}

Each point on the frontier represents a distinct strategy for an investor with certain level of risk-aversion.
It  shows the tradeoff between the expected revenues and the standard deviation.
As we can see from this figure, the frontier increases along an approximate smooth concave curve.
The slope of the ``tangent line'' indicates the trader's risk aversion level $\lambda$.

The point on the most right is  obtained for a risk-neutral trader  with $\lambda=0$. 
We define this point as $  (Std_0, R_0)$.
For any other point on the left, we have
$$
  R_t-R_0\approx\frac{1}{2}(Std_t-Std_0)^2\frac{d^2R}{d   Std^2} \Big|_{Std=Std_0}.
$$
A crucial insight is that for a risk neutral trader, a first-order decrease in the expected revenue can  approximately incur a second-order  decrease in the standard deviation.
The efficient frontier depicted in Figure 2 is consistent with the one for the mean-variance portfolio selection.

\subsubsection{Effect of Illiquidity}

The cost of trading has three main components, brokerage commission, bid-ask spread 
(we will discuss this factor later) and price impact, of which liquidity affects the latter two.
Brokerage commission is usually negotiable and does not constitute a large fraction of the total cost of trading except in small-size trades, which is a topic beyond the scope of this paper.
Stocks with low liquidity can have wide bid-ask spreads.
The bid-ask spread, which is the difference between the buying price and the selling price, is incurred as a cost of trading a security.
The larger the bid-ask spread, the higher is the cost of trading.

Liquidity also has implications for the price impact  of trade.
The extent of the price impact depends on the liquidity of the stock.
A stock that trades millions of shares a day may be less affected than a stock that trades only a few hundred thousand shares a day.
The effect of temporary market impact $  \eta^{tem}$, a key indicator of the illiquidity of
a market, is shown in Figure 3.

It is clear  that traders in a market with larger temporary impact will involve higher  transaction costs; and that the cost  distribution behaves like a normal distribution under
the assumption of constant volatility.
We refer interested readers to the analysis of the effects of market impact, both permanently and temporarily, to the work by Almgren \cite{Almgren3}.

\begin{figure}[H]
\begin{center}
\includegraphics[width=5.in, height=3.4in]{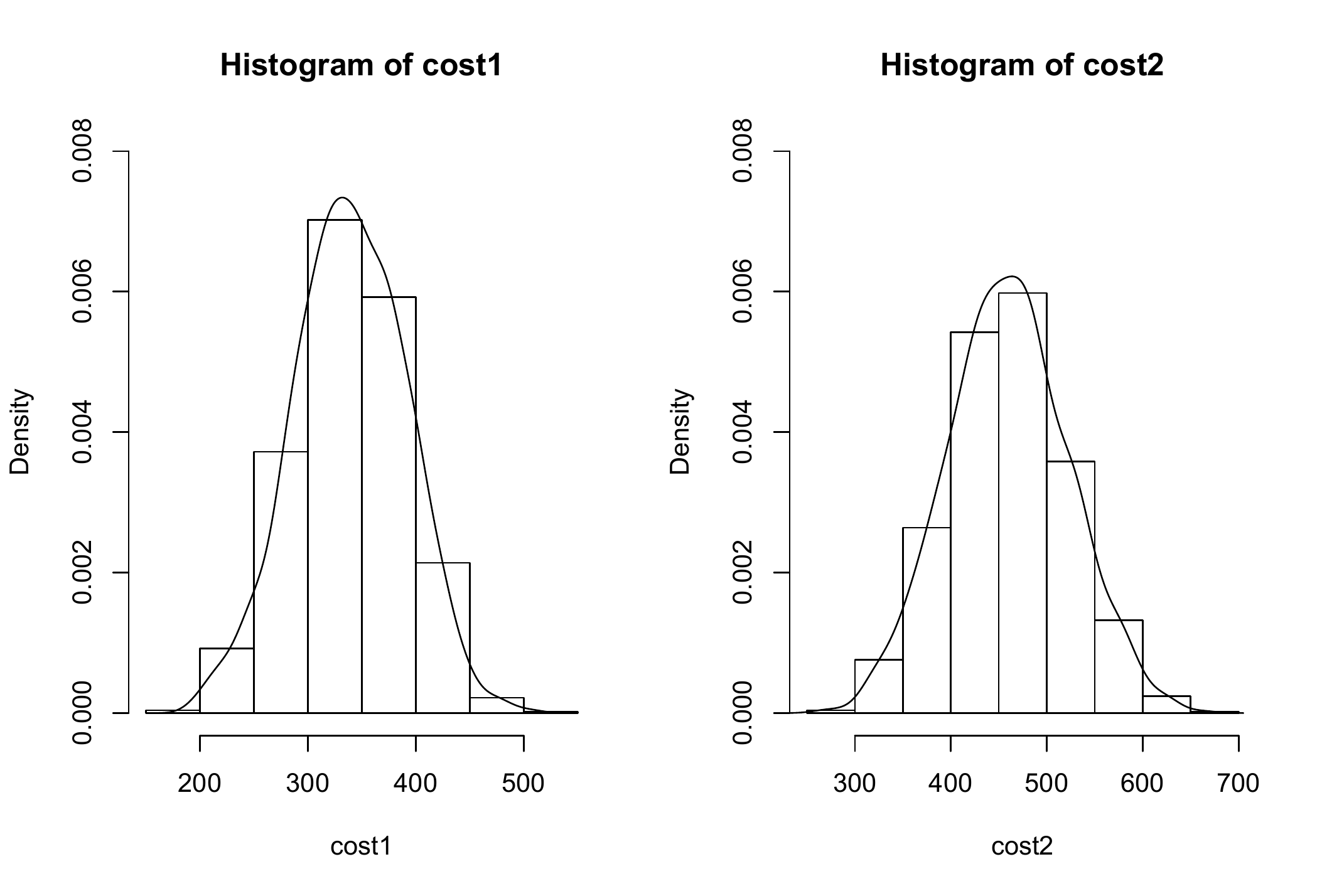}
\label{fig4}
\end{center}
\end{figure}

\begin{figure}[H]
\begin{center}
\includegraphics[width=5.in, height=3.4 in]{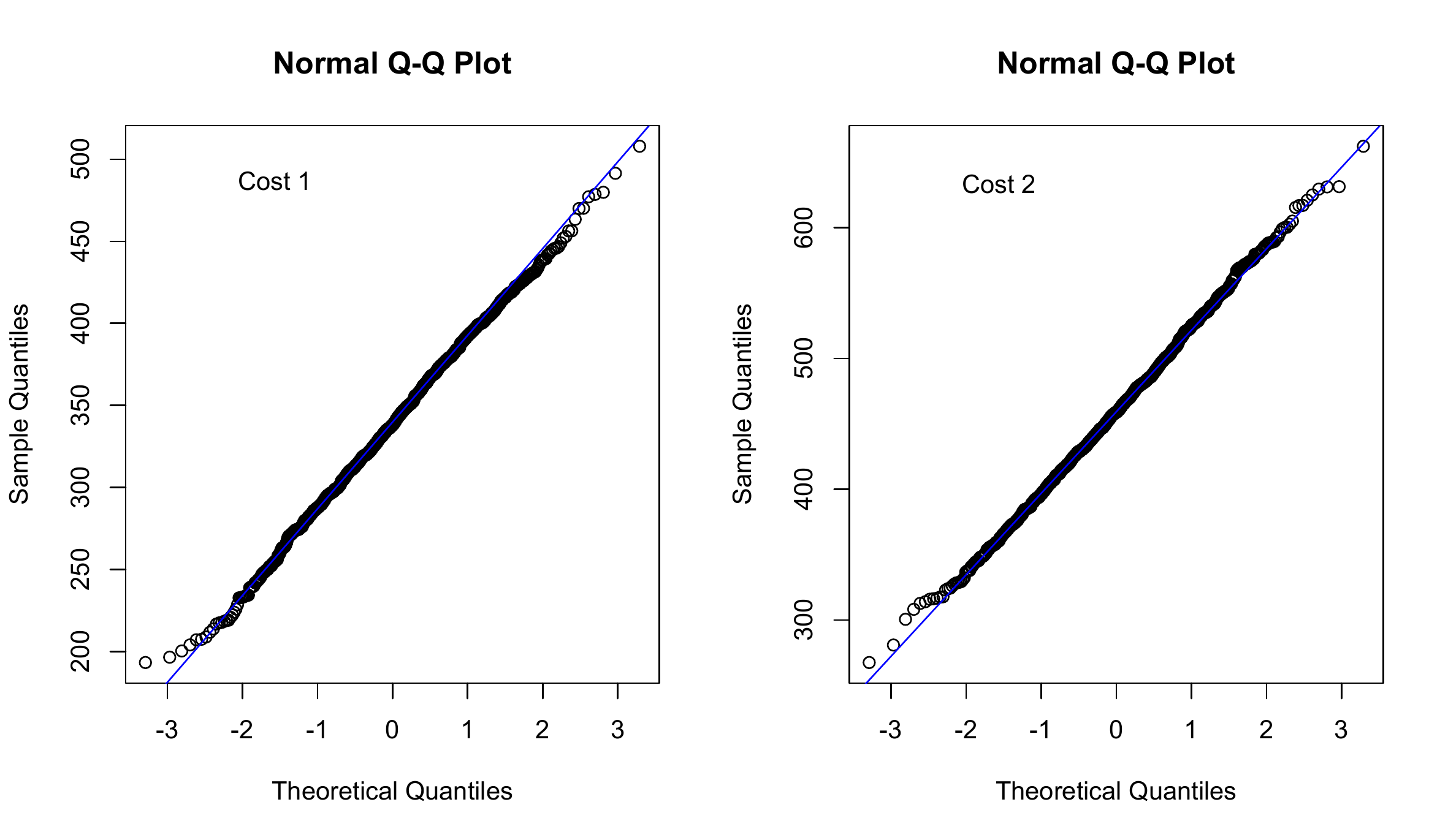}
\caption{    {\bf (Left-hand Side):} $  \eta^{tem}=0.005$;
\quad \quad  {\bf(Right-hand Side):} $  \eta^{tem}=0.01$.  }
\label{fig40}
\end{center}
\end{figure}

\section{Stochastic Volatility Model}

The constant volatility assumption might be reasonable for large-cap US stocks.
However, this assumption becomes defective when we move to the analysis of
small and medium-capitalization stocks, whose volatilities vary randomly through the day.
In this section, we relax the constant volatility assumption.
Thus the value function $J=J(t, s, \sigma, q)$
is a function satisfies HJB-1. 
Similarly, we  consider  relaxing the constraints associated with HJB-1 and solve  the unconstrained optimization problem in  HJB-$1^\prime$. 
We then verify that the obtained optimal control does satisfy all the constraints in HJB-1. 

We present a framework  that is general enough to take into account
the effect of stochastic volatility in the following.

\subsection{Slow Mean-reverting Volatility Model}

We analyze models in which stock prices are conditionally normal\footnote{As we have discussed before, it is reasonable when considering short-term liquidating problem.},
and the volatility process is a positive and increasing function of a mean-reverting 
Ornstein-Uhlenbeck (OU) process. That is,
\begin{equation}\label{stochastic}
\left\{
\begin{array}{lll}
  \displaystyle dS_t=\phi(\nu_t)dW_t\\
  \displaystyle d\nu_t=\epsilon (m -\nu_t)dt + \xi \sqrt{ \epsilon}\; dB_t\\
  B_t=\rho W_t+\sqrt{1-\rho^2}Z_t\\
\end{array}
\right.
\end{equation}
where   $  \{W_t\}$ and $  \{Z_t\}$ are two independent one-dimensional Brownian motions,  and $\rho$ is the correlation between price and volatility shocks, with $ |\rho|<1$.
The parameter $m$ is the equilibrium or long-term mean,
$\epsilon (>0)$ is the intrinsic time-scale of the process, and  $\xi$,  combining with the scalar $\sqrt{\epsilon}$,  represents the degree of volatility around the equilibrium $m$ caused by shocks.
Here $  \{\nu_t\}$ is a simple building-block for a large class of stochastic  volatility models described by choice of $\phi(\cdot)$.
According to  It\^{o}'s lemma,
 $$
   d\phi(\nu_t)=\epsilon\left[\phi^\prime(\nu_t)(m-\nu_t)+\xi^2\phi^{\prime\prime}(\nu_t)\right]dt+\xi \phi^\prime(\nu_t)\sqrt{\epsilon}dB_t.
 $$
We call these models mean-reverting because the volatility is a monotonic function of
$  \{\nu_t\}$ whose drift pulls it towards the mean value $m$.
The volatility is correspondingly pulled towards {\it approximately} $  \phi(m)$.\\

Plenty of analysis on specific It\^{o} models by numerical and analytical methods 
can be found in the literature \cite{RA2, Fouque00}.
Our goal here is to identify and capture the relevant features of liquidating strategies
for small- and medium-capitalization stocks.
Our framework in Eq. (\ref{stochastic}) is adequate and efficient enough to
capture this.
Here we will work in the regime $\epsilon\ll 1$ (slow mean-reverting).
For simplicity, we consider the case of $  N=1$.
The value function $J$ under this setting solves the following PDE:
 \begin{equation}\label{PDE2}
 \left\{
 \begin{array}{lll}
  \displaystyle\left( \partial_t  +\mathcal L  \right)J   -\lambda \phi^2(\nu)q^2
   +\frac{1}{4\eta^{tem}}(\partial_q J+\eta^{per} q)^2=0\\
   J(T-,s, \nu, q)=-Kq^2.
\end{array}
\right.
\end{equation}
Furthermore, we have  $  J(t, s, \nu,q)=J(t, \nu,q)$ independent of $s$,
since the terminal data does not depend on $s$ and PDE (\ref{PDE2})
introduces no $s$-dependence.
Hence, the operator $\mathcal L$ takes the form of
$$
  \mathcal L=\epsilon\left((m-\nu)\partial_\nu+\frac{1}{2}\xi^2 \partial_{\nu\nu}\right).
$$

\noindent 
Similarly as before, we look for a solution quadratic in the inventory variable $q$
\begin{equation}\label{rt}
  J(t,\nu,q)=f(t,\nu)+g(t,\nu) q+h(t,\nu ) q^2.
\end{equation}
The  optimal  liquidating schedule can then   be calculated through the relation
 \begin{equation}\label{rep1}
   \theta^{*}(t,\nu,q)=-\frac{1}{2\eta^{tem}}\big[\big(2h(t,\nu)+ \eta^{per}\big)q+g(t,\nu)\big].
  \end{equation}
To solve Eq. (\ref{rt}),  $f, g$ and $h$  must satisfy the following PDEs:
\begin{equation}\label{poly}
 \begin{array}{lll}
 \left\{
 \begin{array}{l}
  \displaystyle (\partial_t+\mathcal L) h -\lambda\phi^2(\nu)+ \frac{1}{4\eta^{tem}}(2h+\eta^{per})^2=0\\
   \displaystyle h(T-, \nu)=-K
 \end{array}
 \right.\\
 \\
 \left\{
 \begin{array}{l}
  \displaystyle (\partial_t+\mathcal L) g+ \frac{1}{2\eta^{tem}}(2h+\eta^{per})g=0\\
   \displaystyle g(T-,\nu)=0\\
 \end{array}
 \right.\\
 \\
 \left\{
 \begin{array}{l}
   \displaystyle (\partial_t+\mathcal L) f+ \frac{1}{4\eta^{tem}}g^2=0\\
   \displaystyle f(T-, \nu)=0.
 \end{array}
 \right.
 \end{array}
\end{equation}
\quad\\

It is straightforward to verify  that  (Feynman-Kac formula)
$$
 g(t,\nu)\equiv 0\quad \mbox{and} \quad f(t,\nu)\equiv 0.
$$
If we set 
$$ 
 Y_t=\frac{1}{2\eta^{tem}}\int_0^t(2h(u, \nu_u)+\eta^{per})du\quad{\rm and}\quad 
   \mathfrak{h}(t,\nu_t, Y_t)= e^{Y_t} (2h(t,\nu_t)+\eta^{per}).
$$
According to  It\^{o}'s lemma, 
$$
\begin{array}{lll}
  d\mathfrak{h}(t,\nu_t, Y_t)
&=&  \displaystyle e^{Y_t} \Big\{\big[ 2(\partial_t+\mathcal L)h+\frac{1}{2\eta^{tem}}(2h+\eta^{per})^2\big]dt +2\xi\sqrt{\epsilon}\partial_\nu h\; dB_t\Big\}\\
&=&  \displaystyle  e^{Y_t} \left\{2\lambda \phi^2(\nu_t) dt +2\xi\sqrt{\epsilon}\partial_\nu h\; dB_t\right\}.
\end{array}
$$
Applying It\^{o}'s formula to $\mathfrak{h}(u,\nu_u, Y_u)$ between $t$ and $T-$, with $h$ satisfying PDE (\ref{poly}), we get:   
$$
 \mathfrak{h}(t,\nu_t, Y_t)= e^{Y_{T-}}(-2K+\eta^{per})-2\lambda\int_t^Te^{Y_u}\phi^2(\nu_u)du-2\int_t^T\xi\sqrt{\epsilon} \partial_\nu h(u,\nu_u)dB_u.
$$
Hence, 
\begin{equation}\label{sol}
\mathfrak{h}(t,\nu, y)=\mathbb E\left[ e^{Y_{T-}}(-2K+\eta^{per})-2\lambda\int_t^Te^{Y_u}\phi^2(\nu_u)du\Big|\nu_t=\nu, Y_t=y\right].
\end{equation}
If  condition (\ref{cond1}) is satisfied, then $  -2K+\eta^{per}<0$, and hence  $  \mathfrak{h}(t,\nu,y)\le 0$ for any $  (t,\nu,y)\in[0,T)\times \mathbb R\times \mathbb R$.  
Therefore, $  2h(t,\nu)+\eta^{per}\le 0$ for any $  (t,\nu)\in[0,T)\times \mathbb R$. 
Then the control policy in  Eq. (\ref{rep1})    is  optimal. \\

Eq. (\ref{sol}) just  provides  an implicit  analytical  solution to  $\mathfrak{h}$.
However, for the general coefficients $(\phi, m, \epsilon, \xi)$, we do not have an explicit solution.
The small-$\epsilon$ regime gives rise to a regular perturbation expansion in 
the powers of $\epsilon$.
As in \cite{Fouque00}, we seek an asymptotic approximation for $h$:
\begin{equation}\label{ghj}
  h(t,\nu)=h^{(0)}(t,\nu)+\epsilon h^{(1)}(t,\nu)+\epsilon^2 h^{(2)}(t,\nu)+\cdots+.
\end{equation}
This is a series expansion, for which we find $  h^{(i)}$ for $  i=0,1$ explicitly,
and study the accuracy when using the truncated series in Section 4.3.2.  
In the next sections, we present our asymptotic approach for $h$ and the
accuracy of our approximate optimal liquidation policy.

\subsection{Asymptotics}

We first  construct a regular perturbation expansion in the powers of $\epsilon$ by writing
$$
 \mathcal L=\epsilon \mathcal L_0, \quad h(t,\nu)=\sum_{i\ge 0}\epsilon^i h^{(i)}(t,\nu)
$$
where
$$
  \mathcal L_0=(m-\nu)\partial_\nu+\frac{1}{2}\xi^2 \partial_{\nu\nu}.
$$
Substituting these expansions into Eq. (\ref{poly}) and grouping the terms of 
the powers of $\epsilon$, we find that the lowest order equations of the regular perturbation expansion are
   \begin{equation}\label{like}
   \begin{array}{lll}
   \mathcal O(1):&&   \displaystyle \quad\quad\quad\quad0= \partial_t h^{(0)}-\lambda\phi^2(\nu)+\frac{1}{4\eta^{tem}}(2h^{(0)}+\eta^{per})^2 \\
   \mathcal O(\epsilon):&&   \displaystyle \quad\quad\quad\quad0= \partial_t h^{(1)}+ \mathcal L_0 h^{(0)}+\frac{1}{\eta^{tem}} (2h^{(0)}+\eta^{per})h^{(1)}
   \end{array}
   \end{equation}
   with  terminal data
   $$
     h^{(0)}(T-, \nu)=-K\quad \mbox{and}\quad h^{(1)}(T-, \nu)=0.
   $$
The solutions of Eq. (\ref{like}) can be obtained easily.
Indeed, we have
\begin{eqnarray}\label{lm}
\left\{
\begin{array}{lll}
   \displaystyle  h^{(0)}(t,\nu)= \sqrt{ \eta^{tem}\lambda \phi^2(\nu)} \cdot \frac{\varsigma(\nu)e^{-2\sqrt{\lambda \phi^2(\nu)/\eta^{tem}}(T-t)}-1}{\varsigma(\nu) e^{-2\sqrt{\lambda \phi^2(\nu)/\eta^{tem}}(T-t)}+1}-\frac{1}{2}\eta^{per} \\
  \displaystyle h^{(1)}(t,\nu)=\int_t^T  D(r;t) \mathcal L_0h^{(0)}(r,\nu)  dr \\
\end{array}
\right.
\end{eqnarray}
where
$$
 \varsigma(\nu)=\frac{\sqrt{\eta^{tem}\lambda\phi^2(\nu)}-(K-\eta^{per}/2)}{ \sqrt{\eta^{tem}\lambda\phi^2(\nu)}+(K-\eta^{per}/2)}
\quad  {\rm  and} \quad
D(r;t)=\exp\left(\frac{1}{\eta^{tem}}\int_t^r(2h^{(0)}+\eta^{per})d\tau\right).
$$

\subsection{Optimal Strategy}

We now analyze and interpret how the principle expansion terms for the value function can be used in the expression for the optimal liquidation strategy
$\theta^*$, which leads to an approximate feedback policy of the form:
$$
  \theta^*(t,\nu,q)=\theta^{(0),*}(t,\nu,q)+\epsilon \;\theta^{(1),*}(t,\nu,q)+\cdots+.
$$

\subsubsection{Moving-Constant-Volatility  Strategy}

First, we introduce  the zeroth order terms  $  h^{(0)}$
in the expansion for $h$.
This gives the zeroth order optimal liquidation strategy
$$
  \theta^{(0),*}(t,\nu,q):=-\frac{1}{2\eta^{tem}}\left(2h^{(0)}(t,\nu) +\eta^{per}\right)q.
$$
Recall that, in the case of constant volatility approach (Section 3.1),
\begin{equation}\label{cons}
  \theta^{Cons}(t,q;\sigma)=-\frac{1}{2\eta^{tem}}\left(2h^{(0)}(t,\nu^{Cons})) +\eta^{per}\right)q,
\end{equation}
 where $  \nu^{Cons}=\phi^{-1}(\sigma)$ is a constant  over $  [0,T)$. This approximation  might be    reasonable    for most   large-cap US stocks,  for which, over short-term time frames like one day or less,  the volatility can be regard as a constant. However,  this might  not be a reasonable approximation for  assets that are less heavily traded than large-cap US stocks.

The  naive constant-approach liquidation strategy for small- and medium-capitalization stocks  would adopt the strategy for large-cap US stocks, $  \theta^{Cons}$, but with the coefficient $\sigma$ driven by a varying factor $\nu$, i.e.,
\begin{equation}\label{moving}
  \theta^{Cons}(t, q; \phi(\nu))=\theta^{(0),*}(t,\nu,q).
\end{equation}
We call this strategy  the {\it moving constant-approach liquidation strategy}.
This strategy is a successful one if the parameters remains constant when they are estimated from updated segments of historical data.
However, this is not always the case as the risk environment is dynamic.

\subsubsection{First-order Correlation to Optimal Policy}

The approximation to the optimal strategy can be more accurate by going 
into the higher order terms.
Substituting the expansion of $h$ up to terms in $\epsilon$ gives
$$
  \theta^*(t,\nu,q)=\widetilde\theta^*(t,\nu,q)+\mbox{higher order terms}
$$
where
\begin{equation}\label{po23}
  \widetilde\theta^*(t,\nu,q)=-\frac{1}{2\eta^{tem}}\left(2h^{(0)}(t,\nu) +\eta^{per}\right)q-  \frac{\epsilon}{\eta^{tem}}h^{(1)}(t,\nu)q,
\end{equation}
and $  h^{(0)}$ and $  h^{(1)}$ are given by Eq (\ref{lm}).

In Eq.(\ref{po}), the term with  $\epsilon$ corresponds to traders' response to 
the risk arising from  stochastic volatility, which can be regarded as the principle term
hedging the risk factor.
In the previous sections, we presented a ``pure'' stochastic volatility approach to the liquidating problem, in which volatility $  \sigma_t$ is modeled as an It\^{o} process driven by a Brownian motion that has a component independent of the Brownian motion driving the asset price.
In comparison with Almgren's work in \cite{RA2}, we mainly focus on a special class of   volatility model, time-scale volatility model, and worked in the regime $\epsilon\ll 1$.  The separation of time-scales with respect to this kind of models provided an efficient way to identify and track the effect of randomly varying volatility.
The theorem below justifies the accuracy of this approximation.

\begin{theorem}
For fixed $  t<T$, $\nu$ and $  q$, the optimal liquidating strategy $  \theta^*(t,\nu,q)$ defined in Eq. (\ref{rep1}), and our asymptotic approximation
$  \widetilde \theta^*(t,\nu,q)$ defined in Eq. (\ref{po23}) satisfy
$$
  | \theta^*(t,\nu,q) -\widetilde \theta^*(t,\nu,q) |=\mathcal O(\epsilon^2).
$$
\end{theorem}

\begin{proof}
We note that, at any time $  t<T$, according to Eqs. (\ref{rep1}), (\ref{ghj}) and (\ref{po23})
$$
  |\theta^*(t,\nu,q)-\widetilde\theta^*(t,\nu,q)|=\frac{|h-(h^{(0)}+\epsilon\;h^{(1)})|}{\eta^{tem}}q.
$$
Defining the residual $  R^\epsilon= h -( h^{(0)}+\epsilon\; h^{(1)})$.
It solves the following PDE:
$$
\left\{
\begin{array}{lll}
   \displaystyle (\partial_t+\mathcal L)R^\epsilon+ \frac{1}{\eta^{tem}}R^\epsilon[h+ h^{(0)}+\epsilon h^{(1)}+\eta^{per}]+\epsilon^2
\left\{\mathcal L_0h^{(1)}+\frac{1}{\eta^{tem}}(h^{(1)})^2\right\}=0\\
   R^\epsilon(T-,\nu)=0.
\end{array}
\right.
$$
Directly applying the Feynman-Kac formula to this equation  yields
\begin{equation}\label{hj}
  \displaystyle  R^\epsilon(t,\nu)=\epsilon^2\cdot \mathbb E_t \int_t^TD_2(r; t)\{ \mathcal L_0h^{(1)}+\frac{1}{\eta^{tem}}(h^{(1)})^2\} dr
\end{equation}
where
$$
 D_2(r; t)=
\exp\left(\frac{1}{\eta^{tem}}\int_t^r[h+ h^{(0)}+\epsilon h^{(1)}+\eta^{per}]d\tau\right).
$$
One can see by direct computation that the integrand in Eq. (\ref{hj}) is
bounded in $  [0,T]$, i.e., there exists a constant $  C$, such that
$  |R^\epsilon(t,\nu)|\le C\epsilon^2$ which completes our proof.
\end{proof}
\quad\\

We remark that,  given  an initial value $\nu_0=\nu$,
$$
  \nu_t\big|\nu \sim \mathcal N\Big(m+(\nu-m)e^{-\epsilon  t}, \;\frac{\xi^2}{2 }(1-e^{-2\epsilon   t})\Big),
$$
where $  \mathcal N(\mu,\sigma^2)$ is the  Gaussian distribution with mean $\mu$ and standard variation $\sigma$.

Assume that
$$
  \sigma_t=\phi(\nu_t) \quad {  or} \quad \nu_t=\phi^{-1}(\sigma_t)
$$
is observable in real time with some reasonable degree of confidence.
Indeed, we cannot observe directly the volatility, what we may observe is the Volatility Index (VIX) which may be used as a proxy for $  \sigma_t$.
There are different techniques to estimate the parameters $  (\epsilon, m,\xi)$.
These estimators rely on the persistence of market properties (volatility and liquidity), so that information about the past provides reasonable forecasts for the future. Such persistence, at least across short horizons, is well documented (see, for instance,  Bouchaud et al.   \cite{JP} for more details).

\subsection{Simulation Results}

In this section, we assume $  \phi(\nu)=e^{\nu}$ to test the performance of
three different strategies on small- and medium-capitalization stocks:
(i) constant volatility approach in Eq. (\ref{cons});
(ii) moving constant volatility approach in Eq. (\ref{moving}); and
(iii) first-order correction in Eq. (\ref{po}).
We  refer to these strategies as
``(i) Constant-Vol'', `` (ii) Moving-Constant-Vol" and  ``(ii) Vol-adjust'', respectively.

As far as our  simulation is concerned, we use the following hypothetical values of the model parameters:
 $  \nu=0.5, \ m=1, \ \epsilon=0.01, \ \xi=2, \  \rho=-0.4$\footnote{A negative value for $\rho$ is used here to capture the leverage effect.}.
The rest of the parameters are assumed to be the same as those used in Section 3.2.
The simulation is obtained through the following procedure:
at time t, the trader's trading rate $  \theta_t$ is computed, given the state variables. At time $  t+dt$,
the  mid-price is updated by a random increment $  \pm e^\nu\sqrt{dt}$.
The  volatility is updated accordingly by a random increment:
$$
  \epsilon(m-\nu)dt+\xi\sqrt{\epsilon}(\rho\sqrt{dt}\pm \sqrt{1-\rho^2}\sqrt{dt} )
\;\; \mbox{or} \;\;
\epsilon(m-\nu)dt+\xi\sqrt{\epsilon} (-\rho\sqrt{dt}\pm \sqrt{1-\rho^2}\sqrt{dt} ).
$$
Figure \ref{fig5} illustrates the dynamics of  the volatility for one simulation path.
From the simulation path depicted in Figure \ref{fig5}, we can see the leverage effect between the price and the volatility.
\begin{figure}[H]
\begin{center}
\includegraphics[width=5in, height=4.5in]{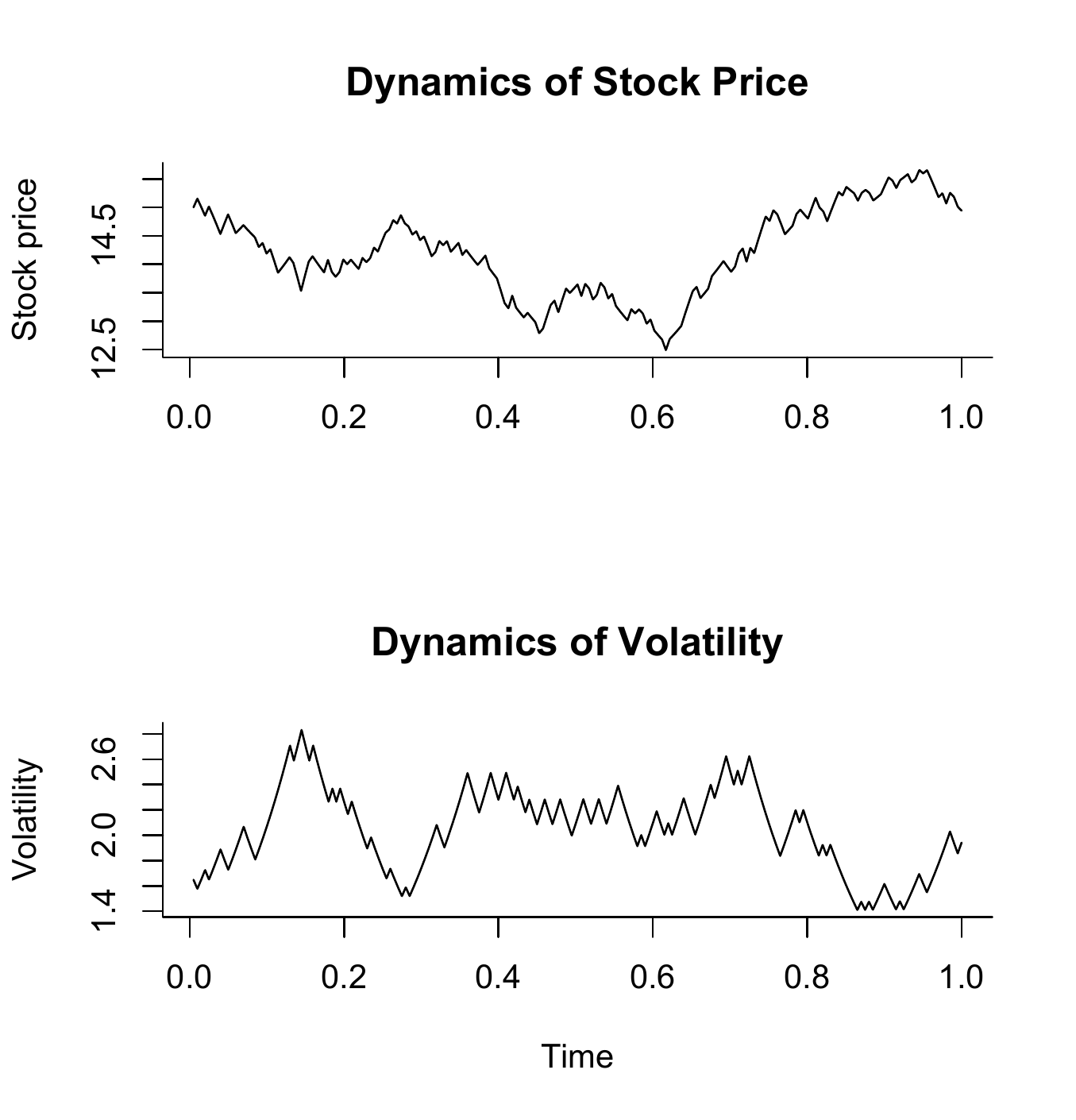}
\caption{ Dynamics of $  (S_t, \exp(\nu_t))$ with $\rho=-0.4$. }
\label{fig5}
\end{center}
\end{figure}

We then run $1,000$ simulations to compare the performance of these  strategies, primarily focusing on the shape of the profit and loss (P\&L) profile.
  The table below shows the final results.
 \begin{table}[H]
\caption{1,000 simulations with $\lambda=0.1$ w.r.t. $  \mathcal R_T$.}
\begin{center}
\begin{tabular}{lccc}
 \hline
\cline{1-4}
Statistics &Constant-Vol  & Moving-Constant-Vol  &Vol-adjust \\
\hline
\cline{1-4}
 Mean &-300.70&-294.50&-288.46\\
 Std&54.76&27.67&27.45\\
 Skewness&1.03&0.23&0.24\\
 Kurtosis&5.26&3.03&3.05\\
\hline
\cline{1-4}
Objective function   &-577.58&-565.39&-560.36\\
\hline
\cline{1-4}
\end{tabular}
\end{center}
\label{t2}
\end{table}%

\begin{figure}[H]
\begin{center}
\includegraphics[width=5in, height=3.8in]{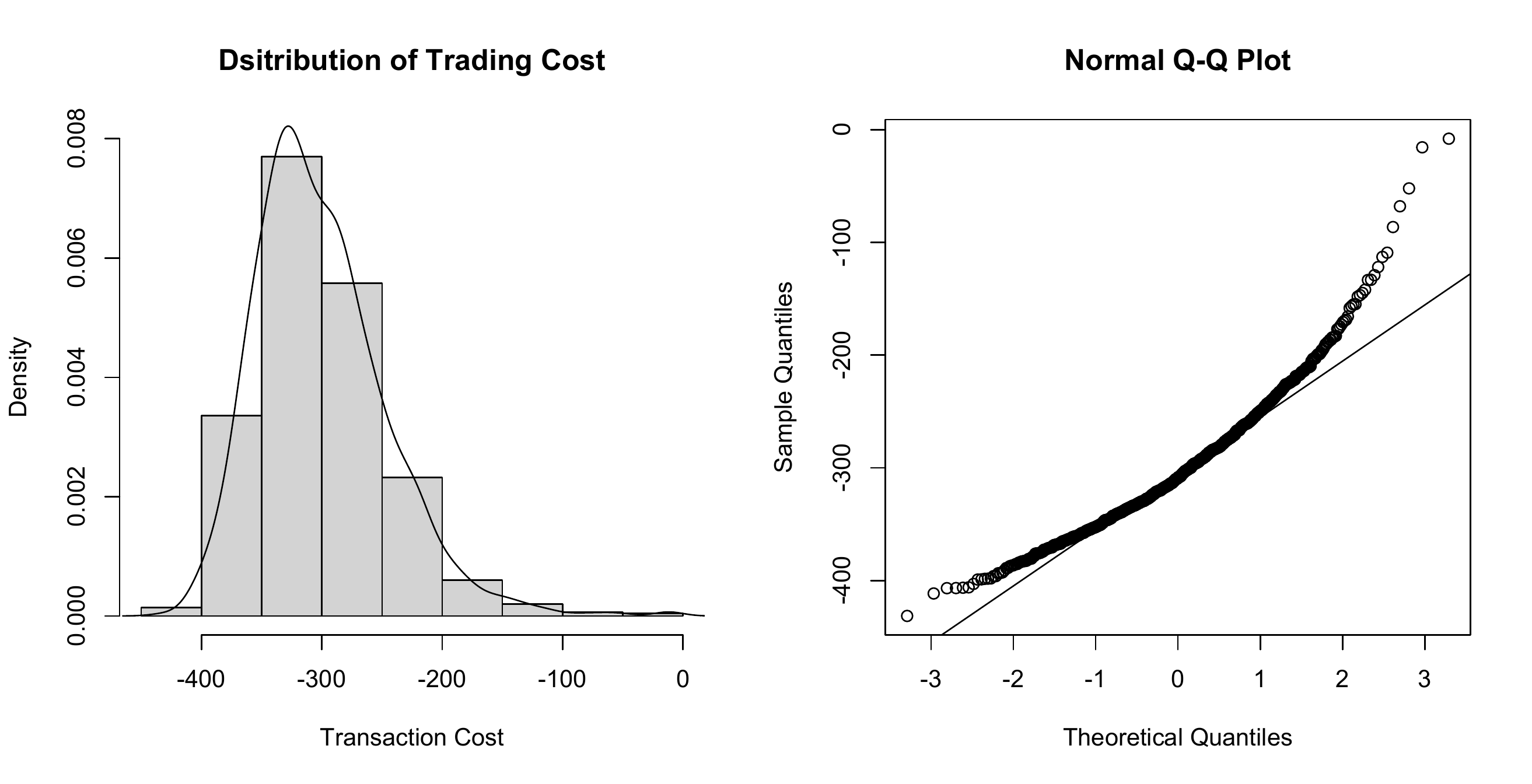}
\caption{  The performance of the ``Constant-Vol" strategy on small- and medium-capitalization assets with their volatilities  varying randomly through the day.    }
\label{fig6}
\end{center}
\end{figure}

As we can see from Table \ref{t2}, the differences among these strategies are significant.  Profiting  from hedging the risk arising from stochastic volatility  is possible. 
Compared with the ``Constant-Vol'' strategy, the results   in the second column,  other strategies adapting to the varying volatility, either ``Moving-Constant-Vol'' (the results in the third column) or ``Vol-adjust'' (the results in the fourth column), obtain a  lower cost and a  smaller risk.  
The more accurate the adaption, the lower the   cost  of trading and  the smaller the  risk of liquidating.

The ``Constant-Vol" strategy is a successful one if the underlying assets are
large-cap US stocks.
However, it is not  applicable for assets that are less heavily traded than large-cap US stocks, where  volatilities are believed to vary randomly through the day.
Figure \ref{fig6} displays  the performance of  the ``Constant-Vol" strategy on small- and medium-capitalization assets.
As we can see from this figure, the distribution of its P \& L  profile is significantly   different from a normal distribution (which is the case for large-cap US stocks)  with quite high levels of skewness and kurtosis
due to the presence of stochastic volatility effect.
Thus at least some adjustment has to be made to identify and capture  the relevant characteristics for these assets.\\

\section{An  Extension: The Incorporation  of Limit Orders }

Consider a  trader, he/she may submit sell limit  orders specifying the quantity and the indicated price per share, which will be executed only when incoming buy market orders are matching with the limit orders.
Otherwise, he/she can post sell market orders for an immediate execution,
but in this case he/she is going to obtain a less favorable execution price.

Suppose that sell limit orders are submitted at the current  best ask, and that the ask-bid spread $\Delta$ keeps unchanged throughout the trading horizon.
Moreover, assume that other participants' buy   market orders arrive at Poisson  rate
$ \lambda^M=A\exp(-\kappa\Delta)$, a decreasing function of $\Delta$ as in the literature.
Let $  M_t$ count  the numbers of buy market orders arriving at the system by time $t$.
The number of shares held by the trader, $ X_t$, can  then be expressed by ($N=1$)
\begin{equation}\label{shares}
X_t=X_0-M_t-\int_0^t\theta_udu,\quad X_0=Q.
\end{equation}
\begin{remark}
Actually, controls belonging to $  \Theta_t$ need to satisfy one more condition:
\begin{equation}\label{condition}
  \theta_t\Big|_{dM_t=1}=0,
\end{equation}
i.e., market order can only be used at times when sell  limit  orders are not executed.
Since for any admissible control process, the integrands in the following two integrals
$$
 \int_0^t\mathbb I_{\{dM_u=0\}}\theta_udu\quad\mbox{and}\quad  \int_0^t\theta_udu
$$
differ at only finitely many times, and when we integrate with respect to the Lebesgue measure $  du$, these differences do not matter. Technically speaking, the set of all jumps of the integrand is a countable set, and hence, it  has Lebesgue measure zero.  Therefore, no mandatory requirement is made to satisfy condition (\ref{condition}) in a  continuous-time  liquidation setting.
\end{remark}

The advantage of market orders over limit orders  is  that they are less affected by averse price impacts. Suppose   a company revises its growth estimates upwards. Those investors who observe piece of this news may react by submitting buy market orders,   executed against sell  limit  orders in the  limit order book that have not yet been withdrawn. All these executed sell  limit orders may incur opportunity costs due to the adverse selection risk.
Market orders, however, will involve unfavorable  transaction fees--market impact. \\

\subsection{Adverse Selection}

It is possible to incorporate the effect of adverse selection in  the trading strategy by assuming that when other participants'   market orders arrive, 
they may induce a jump in the stock price in the direction of the trade.

Assume that  the bid-price\footnote{Since sell market orders are executed against buy   limit orders, $S_t^I$ is actually the best bid at time $t$.} 
satisfies the stochastic differential equation
\begin{equation}\label{Mid}
 d S^I_t=\displaystyle   dS_t+\eta^{u}dM_t-\eta^dd\hat M_t -\eta^{per}\theta_td t,
\end{equation}
where  $  \{\hat M_t\}_{t\ge0}$ is a jump  process satisfying 
\begin{equation}\label{po}
\left\{
\begin{array}{lll}
 \mathbb P( d\hat M_t=1|dM_t=0)=\frac{\lambda^M\eta^u}{\eta^d} dt\\
 \mathbb P( d\hat M_t=1|dM_t=1)=0,
\end{array}
\right.
\end{equation}
and $ \eta^u, \eta^d>0$ represent the jump in bid-price just after a market order event.   

Let us give a brief comment on Eqs. (\ref{Mid}) and (\ref{po}).  
If bid-side liquidity is found for the order at time $  t\in[0,T)$, i.e., $dM_t=1$, the expected execution  price for  the order at time $t$ moves up (by $ \eta^u$). 
If no liquidity is found on that side, i.e., $ dM_t=0$, 
the expected price move is supposed to be  in the opposite direction. 
By Eq. (\ref{po}), 
\begin{equation}\label{po1}
\begin{array}{lll}
 \displaystyle \mathbb P(d\hat M_t=1)&=& \displaystyle \mathbb P(d\hat M_t=1|d M_t=0)\times \mathbb P(dM_t=0)+\mathbb P(d\hat M_t=1|d M_t=1)\times \mathbb P(dM_t=1)\\
&=&  \frac{\lambda^M\eta^u}{\eta^d} dt \times \mathbb P(dM_t=0)\\
&=& \frac{\lambda^M\eta^u}{\eta^d} dt+o(dt), 
\end{array}
\end{equation}
and 
$$ 
\mathbb P(d\hat M_t=0)= 1-\frac{\lambda^M\eta^u}{\eta^d} dt+o(dt).
$$ 
Therefore, 
$$
\begin{array}{lll}
  \mathbb E_t[\Delta S_t^I]&=&  \mathbb E_t[ S_t^I-S_{t-}^I]\\
&=&  [\lambda^M dt+o(dt)]\times \eta^u -[1-\lambda^Mdt+o(dt)]\times [\frac{\lambda^M\eta^u}{\eta^d} dt+o(dt)]\times \eta^d \\
&=&  o(dt).
\end{array}
$$

To  evaluate the effect of the adverse selection on trader's trading strategy, 
we consider the simplest case: risk-neutral investors.  
Similar conclusions may also hold for risk-averse investors.
  \begin{itemize}
\item For a market-order only trader who does not have any limit orders in the market and simply uses market orders to implement his/her liquidation mandate, we have $ \eta^u=\eta^d=0$. That is, the trader does not  use limit orders to  ``monitor" other participants' trading behavior, and believes that that stock price evolves over time   according  to
 $$
  dS^I_t=dS_t-\eta^{per}\theta_tdt.
 $$
 It is straightforward to verify that the solution to the associated  HJB equation is
 $$
   J(t, q)=\left(\frac{1}{\frac{2}{\eta^{per}-2K}-\frac{1}{\eta^{tem}}(T-t)}-\frac{\eta^{per}}{2}\right)q^2.
 $$
 The optimal liquidating strategy can then be derived through the relations,
 \begin{equation}\label{ac}
 \begin{array}{ll}
    \displaystyle \theta^M_t=-\frac{1}{2\eta^{tem}}\left[\partial_q J+\eta^{per}q\right],\\
    \displaystyle X_T=\frac{\eta^{tem}Q}{\eta^{tem}+T\left(2K-\eta^{per}\right)}.\\
 \end{array}
 \end{equation}
It is straightforward to verify that if $K$ satisfies Condition (\ref{cond1}), 
then $\theta^M_t\ge0, X_T\in(0,Q)$, and as $  K\to\infty$, $  X_T\to0$. 
\end{itemize}
It is the simplest liquidating strategy in limit order books for impatient traders
(recall the urgency for liquidation is the primary consideration).

Now let us take limit orders into consideration.
In addition to  earning ask-bid spreads, limit orders can also be used to incorporate the effect of adverse selection on the trading strategy.
  \begin{itemize}
\item For a risk-neutral trader who seeks  to maximize expected terminal wealth with limit and market orders. 
Under some smoothness assumptions for a classical solution or the situation when a unique viscosity solution is considered, the  value function $J$ 
without any constraints   solves
   \begin{equation}\label{case2}
   \left\{
\begin{array}{lll}
  \displaystyle  \partial_t J  
  +\lambda^M \left[\eta^u+\Delta+J(t,q-1)-J(t,q)\right]

+   \frac{1}{4\eta^{tem}}\left[\partial_q J+\eta^{per}q\right]^2=0\\
 \displaystyle J(T-, q)=-K q^2.
\end{array}
\right.
\end{equation}
 We conjecture that the solution has the following form:
$$
 J(t,q)= f(t)+g(t)q+h(t)q^2.
$$
Some simple calculations lead to:
\begin{equation}\label{relation}
\left\{
\begin{array}{lll}
 \displaystyle  h(t)=\frac{1}{\frac{2}{\eta^{per}-2K}-\frac{1}{\eta^{tem}}(T-t)}-\frac{\eta^{per}}{2}\\
 \displaystyle g(t)=-2\lambda^M\int_t^Te^{\frac{1}{2\eta^{tem}}\int_t^r(2h(\tau)+\eta^{per})d\tau}h(r)dr\\
 \displaystyle f(t)=\int_t^T\Big[\lambda^M\big(\eta^u+\Delta+h(r)-g(r)\big) +\frac{g^2(r)}{4\eta^{tem}} \Big]dr,
\end{array}
\right.
\end{equation}
and
\begin{equation}\label{appo}
   \displaystyle  \theta^{M\&L}_t=\displaystyle  -\frac{1}{2\eta^{tem}}\left[g(t)+\left(2h(t)+\eta^{per}\right)q\right].
\end{equation}
 Here $  \theta^{M\&L}_t$ denotes the optimal liquidating strategy with limit and market orders.
\end{itemize}
\quad\\

We remark that,  
\begin{description}
\item[(1)] For any    $  t\in[0,T)$,    $  h(t)<0$, $  g(t)>0$. 
Therefore, 
$$
 \theta^{M\&L}_t <\theta_t^M.
$$
That is, realizing the information carried by other marker participants 
(as we assumed in our paper, it is monitored by the execution of limit orders), 
traders will  slow down their liquidating speed so as to get profit from market momentum.

\item[(2)] If   condition (\ref{cond1}) is satisfied, 
$$
\begin{array}{lll}
  \displaystyle g(t)&=&  \displaystyle -2\lambda^M\int_t^Te^{\frac{1}{2\eta^{tem}}\int_t^r(2h(\tau)+\eta^{per})d\tau}h(r)dr\\
&=&  \displaystyle -2\eta^{tem} \lambda^M\bigg\{  \int_t^Te^{\frac{1}{2\eta^{tem}}\int_t^r(2h(\tau)+\eta^{per})d\tau}d\left(\frac{1}{2\eta^{tem}}\int_t^r(2h(\tau)+\eta^{per})d\tau\right)\\
&&  \displaystyle \quad\quad\quad \quad \quad-\frac{\eta^{per}}{2\eta^{tem}}  \int_t^Te^{\frac{1}{2\eta^{tem}}\int_t^r(2h(\tau)+\eta^{per})d\tau}dr\bigg\} \\
&=& \displaystyle -2\eta^{tem}\lambda^M\left\{ e^{\frac{1}{2\eta^{tem}}\int_t^T(2h(\tau)+\eta^{per})d\tau} -1-\frac{\eta^{per}}{2\eta^{tem}}  \int_t^Te^{\frac{1}{2\eta^{tem}}\int_t^r(2h(\tau)+\eta^{per})d\tau}dr\right\}.
\end{array}
$$
Combining the results in Eq. (\ref{relation}), we have 
$$
   e^{\frac{1}{2\eta^{tem}}\int_t^r(2h(\tau)+\eta^{per})d\tau}=\frac{\frac{1}{2K-\eta^{per}}+\frac{1}{2\eta^{tem}}(T-r)}{\frac{1}{2K-\eta^{per}}+\frac{1}{2\eta^{tem}}(T-t)}.
 $$
If  the time horizon $T$ satisfies 
 $$
   T\le \frac{2\eta^{tem}}{2K-\eta^{per}}\left[\sqrt{1+\frac{2(2K-\eta^{per})}{\eta^{per}}}-1\right],
 $$
then $ g(t)\le 2\eta^{tem}\lambda^M$    holds for any time $  t\in[0,T)$. 
 That is, the reduced liquidation speed  $\frac{g(t)}{2\eta^{tem}}$ is always smaller 
than the rate at which limit sell orders are  executed.   
The overall expected liquidation speed  has increased. 
\item[(3)]  
$$
\left\{
\begin{array}{lll}
 \displaystyle\mathbb E[\dot{X}_t]= -\lambda^M-\mathbb E[\theta^{M\&L}_t]
=  -\lambda^M+\frac{1}{2\eta^{tem}}\left[g(t)+\left(2h(t)+\eta^{per}\right)\mathbb E[X_t] \right]\\
\displaystyle X_0=Q,
\end{array}
\right.
 $$
which yields 
 $$
\mathbb E[X_t]
=Qe^{\frac{1}{2\eta^{tem}}\int_0^t(2h(\tau)+\eta^{per})d\tau}-\int_0^t
\left[\lambda^M-\frac{1}{2\eta^{tem}}g(u)\right]e^{\frac{1}{2\eta^{tem}}\int_u^t(2h(\tau)+\eta^{per})d\tau}du.
 $$
The liquidation target  
 $$
  Q\ge \max_{t\in[0,T)} e^{-\frac{1}{2\eta^{tem}}\int_0^t(2h(\tau)+\eta^{per})d\tau}
\left[\int_0^t\left[\lambda^M-\frac{1}{2\eta^{tem}}g(u)\right]e^{\frac{1}{2\eta^{tem}}\int_u^t(2h(\tau)+\eta^{per})d\tau}du-\frac{g(t)}{2h(t)+\eta^{per}}\right],
  $$
and we have $ \mathbb E[\theta_t^{M\&L}]\ge0$ and $  X_t\in[0,Q]$ for any time  $  t\in[0,T)$. 
The obtained optimal strategy in Eq. (\ref{appo}) is  actually  the one we are looking for. 
\end{description}
 \quad\\
 
The estimation procedure for the parameter $  \lambda_M$ in Eq. (\ref{relation}) basically matches the intuition that one must count the number of executions at ask/bid  and normalized this quantity by the length of the time horizon. For high-frequency data over $  [0,T]$, we have a consistent estimator of $  \lambda_M$, which is given by
$$
  \widehat{\lambda}_M=\frac{\mbox{\#\{ executions at ask/bid\}}}{T}.
$$

\subsection{Numerical Results }

Consider the situation in which stocks are traded in a single exchange, i.e., $  N=1$.
As far as our simulation is concerned, we adopt the hypothetical values of the model parameters $  \lambda_M=100$, $\Delta=0.3$ and $  \eta^u=\eta^d=\eta^I=0.02$.
Other parameters not listed here are assumed to be the same as those used in 
Section 3.2.  
The simulation is obtained through the following procedure:
 \begin{table}[H]
\caption{ The procedure }
\begin{center}
\begin{tabular}{ll}
\hline
\cline{1-2}
Step 0.&Set  initial values at time $  t=0$;\\
Step 1.& Compute the trader's liquidating rate $  \theta_t$, given the state variables; \\
Step 2. & If a buy market order arrives at time $t$, then update the quantities by:\\
&\quad\quad$$
$
\left\{
\begin{array}{l}
  \displaystyle X_{t+dt}=X_t-1\\
  \displaystyle C_{t+dt}=C_t+(S^I_t+\Delta)\\
  \displaystyle S^I_{t+}=S^I_t+\eta^I
\end{array}
\right.
$
$$\quad execute  a trade using  sell LOs;\\
& otherwise, \\
& \quad\quad$$
$
\left\{
\begin{array}{l}
  \displaystyle X_{t+dt}=X_t-\theta_tdt\\
 \displaystyle C_{t+dt}=C_t+( S^I_t-\eta^{tem}\theta_t)\cdot\theta_tdt
\end{array}
\right.
$
$$\quad execute a trade   using   sell  MOs,\\
&and \\
&\quad $  S^I_{t+}=\left\{
\begin{array}{lll}
 \displaystyle S^I_t-\eta^I-\eta^{per}\theta_tdt, &&\mbox{with probability  $ \lambda^M dt$,}\\
  \displaystyle S^I_t-\eta^{per}\theta_tdt, &&\mbox{with probability  $  1-\lambda^M dt$.}\\
\end{array}
\right.
$\\
Step 3.& Update the affected price by a random increment $  \pm \sigma\sqrt{dt}$:\\
&\quad\quad\quad\;$  S^I_{t+dt}=\left\{
\begin{array}{lll}
  S^I_{t+}+\sigma \sqrt{dt},&&\mbox{with probability  0.5,}\\
  S^I_{t+}-\sigma \sqrt{dt},&&\mbox{with probability  0.5;}\\
\end{array}
\right.
$\\
Step 4.& Let $  t=t+dt$. If $  t< T$, return to Step 1; otherwise, stop and exit.\\
\hline
\cline{1-2}
\end{tabular}
\end{center}
\label{default}
\end{table}%
\noindent

 Figures \ref{fig7}  illustrates  the dynamics of inventory and affected price  for one simulation of a stock path.

\begin{figure}[H]
\begin{center}
\includegraphics[width=5in, height=6.5in]{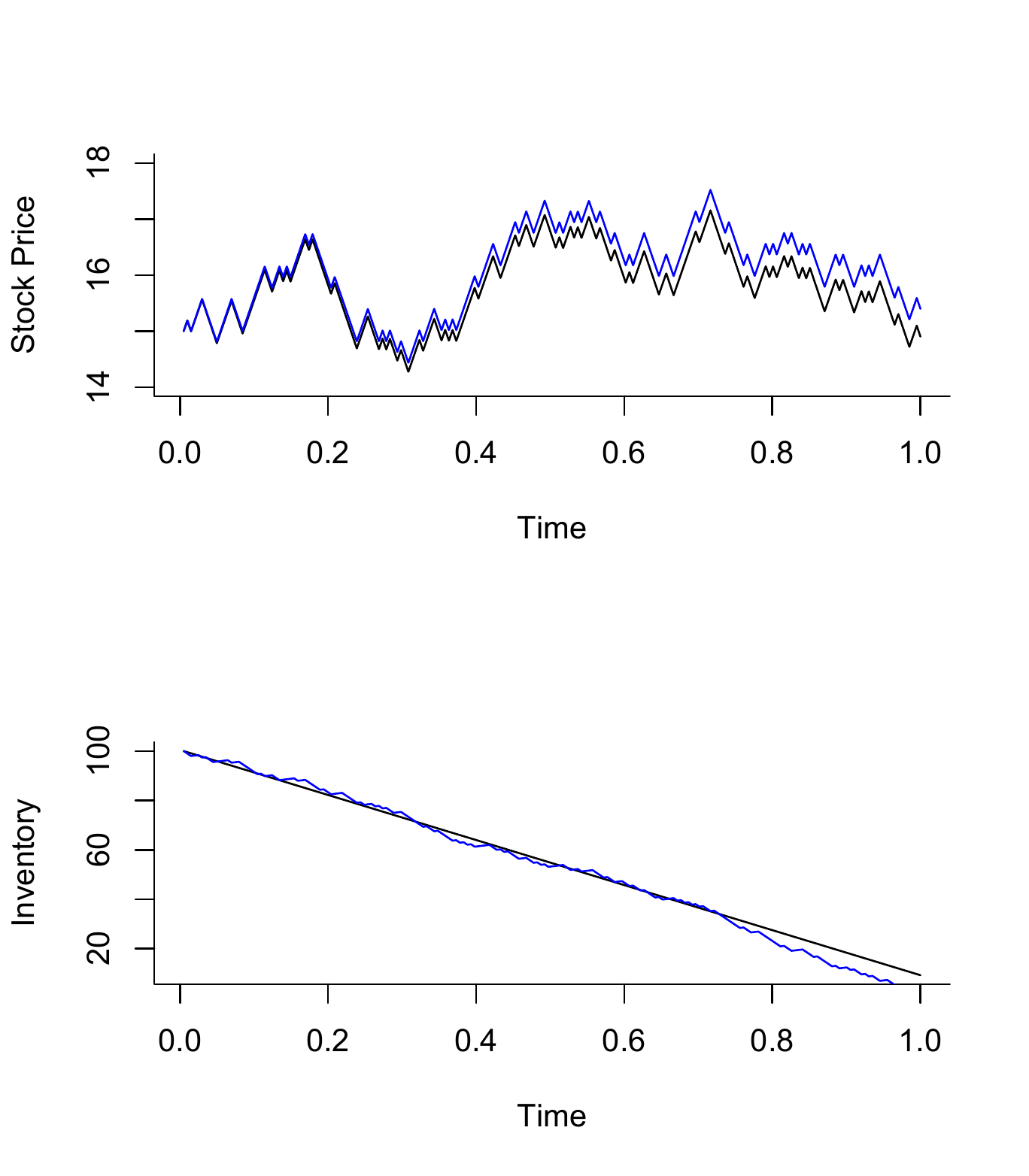}
\caption{ The dynamics of  inventory and  affected price: {(\bf Black}) Market orders only; {(\bf Blue}) Market \& limit.   }
\label{fig7}
\end{center}
\end{figure}

We observe that investors  without making use of  limit orders would   liquidate on a linear trajectory
$  \theta=\frac{q}{a+b(T-t)}$ and receive a relatively  lower execution  price from a certain point in time during the trading horizon.\\

We then run 1000 simulations to investigate the effect of adverse selection on the trading strategy.
The performances of  strategies  aware of this effect and strategies not aware of this effect are depicted in Figure \ref{fig9}.

\begin{figure}[H]
\begin{center}
\includegraphics[width=5 in, height=3.4 in]{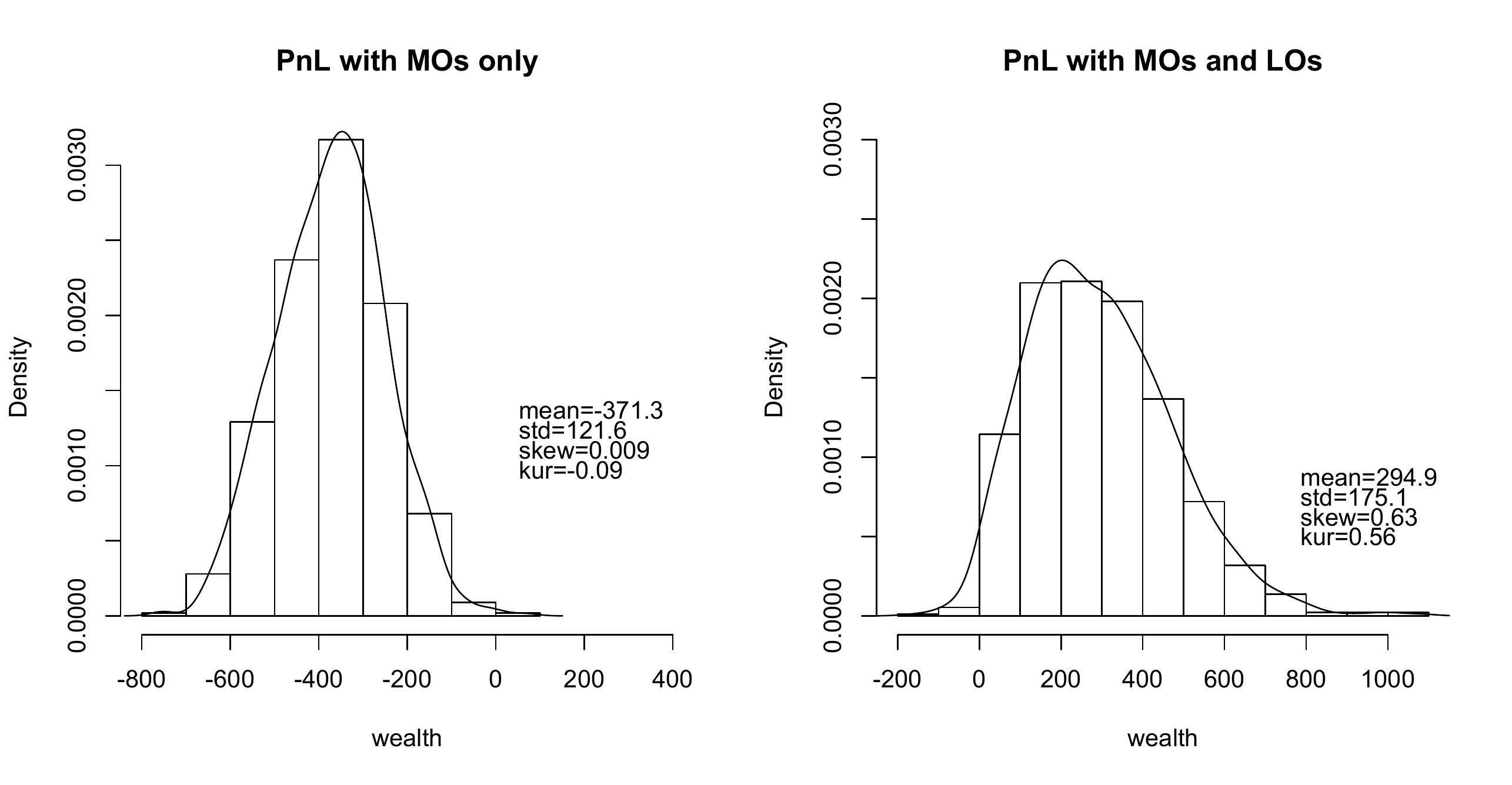}
\caption{ Profit \& Loss profile distributions.  }
\label{fig9}
\end{center}
\end{figure}

The numerical simulations show that the strategy including the use of limit orders performs better than the market only strategy,
it receives a significantly higher expected Profit \& Loss profile
(from risk-neutral traders' point of view),
say $294.9$ with both MOs and LOs and $-371.3$ with MOs only.

\section{Conclusions}

In this paper,  we utilize a quantitative model  to discuss the  optimal liquidating  problem in an illiquid market.
Three different aspects of the optimal liquidation problem are discussed in our paper:
(i) optimal liquidating strategy  across multiple venues;
(ii) optimal liquidating strategy  with stochastic volatility (a special case, slow mean-reverting stochastic volatility, was discussed); and
(iii) the incorporation of limit orders in optimal liquidation  problem.

We formulate the optimal liquidation problem   with both temporary and permanent  market impacts. Under the no arbitrage assumption, we propose  a  linear model  to determine  the equilibrium price in a competitive market with  multiple trading  venues.
Multi-scale analysis method is used to discuss the case of  stochastic volatility, where the stochastic  volatility is driven by a slow-varying factor $  \nu_t$.
As an extension of our model, we include the use of limit orders.
These orders are submitted at the current best ask/bid with buy limit order not being executed.
In addition to earn  ask-bid spreads,  these orders can be used to incorporate the effect of adverse selection on the trading strategy.

In our  model, the quantity of limit sell orders is assumed to be {\it one} share per trade.
Whereas, in reality, it can be multiple shares per trade.
This will be an interesting direction  for future research.

\section*{Appendix}

\subsection*{Appendix A.}

\begin{proof}
In this appendix, we give the steps for deriving HJB-1.  
Using Eq. (\ref{R}),  the objective at any time $t$  becomes
$$
\begin{array}{lll}
\displaystyle
 \mathbb E_t\bigg[\int_t^{T-}\sigma_uX_udW_u-\eta^{per}
\int_t^{T-}X_u(\beta^{(1)}_u\theta^{(1)}_u+\cdots+\beta^{(N)}_u\theta^{(N)} _u)du\\
\displaystyle\quad\quad -\int_t^{T-}[\eta^{tem,(1)}(\theta^{(1)}_u)^2+\cdots
+\eta^{tem,(N)}(\theta^{(N)} _u)^2]du
-\lambda\int_t^{T-}\sigma_u^2X_u^2du -X_{T-}\mathcal C^o(X_{T-})\bigg].
\end{array}
$$
 Thus, given initial value $  X_t=q$,
 $$
 \begin{array}{lll}
J(t,s, \sigma, q)&=&\displaystyle \max_{\boldsymbol\theta(\cdot)\in \Theta_t}\; \mathbb E^{t,s,\sigma, q}_{\boldsymbol\theta(\cdot)}\bigg[\int_t^{T-}\sigma_uX_udW_u-\eta^{per}\int_t^{T-}X_u(\beta^{(1)}_u\theta^{(1)}_u+\cdots+\beta^{(N)}_u\theta^{(N)} _u)du\\
 \displaystyle &&-\int_t^{T-}\left([\eta^{tem,(1)}(\theta^{(1)}_u)^2+\cdots+\eta^{tem,(N)}(\theta^{(N)} _u)^2]- \lambda \sigma_u^2X_u^2\right)du-X_{T-}\mathcal C^o(X_{T-})\bigg]\\
 &=&\displaystyle \max_{\boldsymbol\theta(\cdot)\in\Theta_t}\; \mathbb E^{t,s,\sigma, q}_{\boldsymbol\theta(\cdot)}\Bigg[\int_t^{t+\Delta t}\sigma_uX_udW_u-\eta^{per}\int_t^{t+\Delta t}X_u(\beta^{(1)}_u\theta^{(1)}_u+\cdots+\beta^{(N)}_u\theta^{(N)} _u)du\\
&&  \displaystyle -\int_t^{t+\Delta t}\left([\eta^{tem,(1)}(\theta^{(1)}_u)^2+\cdots+\eta^{tem,(N)}(\theta^{(N)} _u)^2]-\lambda \sigma_u^2X_u^2\right)du\\
&&  \displaystyle +\mathbb E^{t+\Delta t,s+\Delta s, \sigma+\Delta\sigma, q+\Delta q}_{\boldsymbol\theta(\cdot)} \bigg[\int_{t+\Delta t}^{T-}\sigma_uX_udW_u-\eta^{per}\int_{t+\Delta t}^{T-}X_u(\beta^1_u\theta^{(1)}_u+\cdots+\beta^N_u\theta^{(N)} _u)du\\
&&   \displaystyle -\int_{t+\Delta t}^{T-}\left([\eta^{tem,(1)}(\theta^{(1)}_u)^2+\cdots+\eta^{tem,(N)}(\theta^{(N)} _u)^2]-\lambda \sigma_u^2X_u^2\right)du
 -X_{T-}\mathcal C^o(X_{T-})\bigg]\Bigg],\\
\end{array}
$$
where $\mathbb E^{t,s,\sigma, q}_{\boldsymbol\theta(\cdot)}[\cdot]$  is the conditional   expectation conditioned on  the  control process $ \boldsymbol\theta(\cdot)$ and the initial state   $  (S_t, \sigma_t, X_t)=(s,\sigma,q)$.

Notice that for any control process $  \boldsymbol\theta(\cdot)\in\Theta_t$, we have
$$
\begin{array}{lll}
&&  \displaystyle \mathbb E^{t+\Delta t, s+\Delta s, \sigma+\Delta \sigma, q+\Delta q}_{\boldsymbol\theta(\cdot)} \bigg[\int_{t+\Delta t}^{T-}\sigma_uX_udW_u-\eta^{per}\int_{t+\Delta t}^{T-}X_u(\beta^{(1)}_u\theta^{(1)}_u+\cdots+\beta^{(N)}_u\theta^{(N)} _u)du\\
&& \displaystyle -\int_{t+\Delta t}^{T-}\left([\eta^{tem,(1)}(\theta^{(1)}_u)^2+\cdots+\eta^{tem,(N)}(\theta^{(N)} _u)^2]-\lambda \sigma_u^2X_u^2\right)du-X_{T-}\mathcal C^o(X_{T-})\bigg]\\
&=& \displaystyle U(t+\Delta t, s+\Delta s, \sigma+\Delta \sigma, q+\Delta q, \boldsymbol\theta(\cdot))\\
&\le & \displaystyle \max_{\boldsymbol\theta(\cdot)\in\Theta_{t+\Delta t}} U(t+\Delta t, s+\Delta s, \sigma+\Delta \sigma,  q+\Delta q, \boldsymbol\theta(\cdot))\\
&=& \displaystyle J(t+\Delta t, s+\Delta s, \sigma+\Delta \sigma, q+\Delta q).
\end{array}
$$
Thus, we have
\begin{equation}\label{le}
\begin{array}{lll}
  J(t,s,\sigma, q) &\le&  \displaystyle  \max_{\boldsymbol\theta(\cdot)\in\Theta_t}\; \mathbb E^{t,s,\sigma, q}_{\boldsymbol\theta(\cdot)}\Big[J(t+\Delta t,s+\Delta s, \sigma+\Delta\sigma,  q+\Delta q)+\int_t^{t+\Delta t}\sigma_uX_udW_u\\
&& \displaystyle -\eta^{per}\int_t^{t+\Delta t}X_u(\beta^{(1)}_u\theta^{(1)}_u
+\cdots+\beta^{(N)}_u\theta^{(N)} _u)du\\
&&  \displaystyle-\int_t^{t+\Delta t}\left([\eta^{tem,(1)}(\theta^{(1)}_u)^2+\cdots+\eta^{tem,(N)}(\theta^{(N)} _u)^2]-\lambda \sigma_u^2X_u^2\right)du\Big].\\
\end{array}
\end{equation}
Let $ \tilde{\boldsymbol\theta}^*(\cdot)$ be the optimal control over $  [t+\Delta t, T)$, modifying the optimal control between t and $  t+\Delta t$ by an arbitrary control $  \boldsymbol\theta(\cdot)$, i.e. $
  \boldsymbol\theta^\prime_u=\boldsymbol\theta\mathbb I_{\{u\le  t+\Delta t\}}+\tilde{\boldsymbol\theta}^*_{u}\mathbb  I_{\{u> t+\Delta t\}}$.
Then, we obtain
$$
\begin{array}{lll}
  J(t,s,\sigma, q) &\ge&  J_{ \boldsymbol\theta^\prime(\cdot)}(t,s,\sigma, q)\\
&=&  \displaystyle  \mathbb E^{t,s, \sigma, q}_{\boldsymbol\theta(\cdot)}\Big[J(t+\Delta t, s+\Delta s, \sigma+\Delta \sigma,  q+\Delta q)+\int_t^{t+\Delta t}\sigma_uX_udW_u\\
&&  \displaystyle -\eta^{per}\int_t^{t+\Delta t}X_u(\beta^{(1)}_u\theta^{(1)}_u+\cdots+\beta^{(N)}_u\theta^{(N)} _u)du\\
 &&  \displaystyle  -\int_t^{t+\Delta t}\left([\eta^{tem,(1)}(\theta^{(1)}_u)^2+\cdots+\eta^{tem,(N)}(\theta^{(N)} _u)^2]-\lambda \sigma_u^2X_u^2\right)du\Big].\\
\end{array}
$$
Hence, we obtain
\begin{equation}\label{ge}
\begin{array}{lll}
  J(t,s,\sigma, q) &\ge&
   \displaystyle \max_{\boldsymbol\theta(\cdot)\in\Theta_t}\;
 \mathbb E^{t,s,\sigma, q}_{\boldsymbol\theta(\cdot)}\Big[J(t+\Delta t,s+\Delta s, \sigma+\Delta\sigma,  q+\Delta q)+\int_t^{t+\Delta t}\sigma_uX_udW_u\\
 &&  \displaystyle -\eta^{per}\int_t^{t+\Delta t}X_u(\beta^{(1)}_u\theta^{(1)}_u+\cdots+\beta^{(N)}_u\theta^{(N)} _u)du\\
&&  \displaystyle -\int_t^{t+\Delta t}\left([\eta^{tem,(1)}(\theta^{(1)}_u)^2+\cdots+\eta^{tem,(N)}(\theta^{(N)} _u)^2]-\lambda \sigma_u^2X_u^2\right)du\Big].\\
\end{array}
\end{equation}
 Putting both inequalities, Eq.(\ref{le}) and Eq.(\ref{ge}),  together, we arrive at the Dynamic Programming Principle (DPP):
 $$
 \begin{array}{lll}
  J(t,s,\sigma, q)&=&
   \displaystyle \max_{\boldsymbol\theta(\cdot)\in\Theta_t}\;
 \mathbb E^{t,s,\sigma, q}_{\boldsymbol\theta(\cdot)}\Big[J(t+\Delta t, s+\Delta s, \sigma+\Delta\sigma,  q+\Delta q)+\int_t^{t+\Delta t}\sigma_uX_udW_u\\
 &&  \displaystyle  -\eta^{per}\int_t^{t+\Delta t}X_u(\beta^{(1)}_u\theta^{(1)}_u+\cdots+\beta^{(N)}_u\theta^{(N)} _u)du\\
 &&  \displaystyle -\int_t^{t+\Delta t}\left([\eta^{tem,(1)}(\theta^{(1)}_u)^2+\cdots+\eta^{tem,(N)}(\theta^{(N)} _u)^2]-\lambda \sigma_u^2X_u^2\right)du\Big].
\end{array}
$$
Let $  \boldsymbol\theta^*(\cdot)$ be the optimal control over $  [t, T)$, then we have
$$
  \boldsymbol\theta^*_{t,s,\sigma, q}(t^\prime,s^\prime, \sigma^\prime,  q^\prime)=\tilde{\boldsymbol\theta}_{t+\Delta t,s+\Delta s, \sigma+\Delta \sigma,  q+\Delta q}^*(t^\prime, s^\prime, \sigma^\prime,  q^\prime), \quad\mbox{for any  \; $  t^\prime>t+\Delta t$}.
$$
The dynamic programming equation, HJB equation, is the  infinitesimal version of this principle.
\end{proof}

\subsection*{Appendix B.}
\begin{proof}
In the case of (1) $N=1$; or (2)  the trading venues have the same market efficiency, 
$\Delta_N=\lambda \sigma^2>0$. 
Hence, 
$$
  h(t;\sigma)=\sqrt{\frac{\Delta_N}{a}}\cdot \frac{\varsigma e^{-2\sqrt{a\Delta_N}(T-t)}-1}{\varsigma e^{-2\sqrt{a\Delta_N}(T-t)}+1}-\frac{\eta^{per}}{2N}.
$$
 Under the assumption  that 
$ K>\displaystyle  \frac{b}{2a}+\sqrt{\frac{|\Delta_N|}{a}}$, $ -1<\varsigma<0$. 
Therefore, $h$ is a  decreasing function in $t$, and for any $  n=1,\cdots, N$
 \begin{equation}\label{proof1}
 \begin{array}{lll}
   2h(t;\sigma)+\eta^{per}\beta^{(n)}&=&  \displaystyle \sqrt{\frac{\Delta_N}{a}}\cdot \frac{\varsigma e^{-2\sqrt{a\Delta_N}(T-t)}-1}{\varsigma e^{-2\sqrt{a\Delta_N}(T-t)}+1} \le0.
 \end{array}
 \end{equation}
 Recall that 
 $$
 \left\{
 \begin{array}{l}
\displaystyle \dot{X}_t=-\sum_{n=1}^N \theta^{n,*}_t\\
\displaystyle  \theta^{n,*}_t
=-\frac{1}{2\eta^{tem}} \left[2h(t;\sigma)+ \eta^{per}\beta^{(n)}\right]X_t.
 \end{array}
 \right.
 $$
 Hence,  $X_t$ satisfies the following first-order ODE:
 $$
 \left\{
 \begin{array}{l}
 \displaystyle \dot{X}_t=\left( \sum_{n=1}^N\frac{1}{2\eta^{tem}} \left[2h(t;\sigma)
+ \eta^{per}\beta^{(n)}\right]\right)X_t\\
  \displaystyle   X_0=Q,
 \end{array}
 \right.
 $$
 which yields
\begin{equation}\label{proof2}
   X_t=Q\cdot \exp\left(\int_0^t\sum_{n=1}^N\frac{1}{2\eta^{tem}} \left[2h(u;\sigma)+ \eta^{per}\beta^{(n)}\right]du \right) \in (0, Q],\quad \mbox{for} \  t \in [0,T).
 \end{equation}
 Combining the results in Eq. (\ref{proof1}) and Eq. (\ref{proof2}),  we conclude that,  for any $n=1,\ldots, N$,
$$
\theta^{n,*}_t=-\frac{1}{2\eta^{tem}} \left[2h(t;\sigma)+ \eta^{per}\beta^{(n)}\right]X_t\ge 0,
$$
and  
 $$
 \int_0^T\sum_{n=1}^N\theta^{n,*}_tdt=X_0-X_T
=Q\left[ 1-\exp\left(\int_0^t\sum_{n=1}^N\frac{1}{2\eta^{tem}} \left[2h(u;\sigma)+ \eta^{per}\beta^{(n)}\right]du\right) \right]\le Q.
 $$
 
\end{proof}

\section*{Acknowledgements}

This research work was supported by Research Grants Council of 
Hong Kong under Grant Number 17301214, 
HKU Strategic Theme on Computation and Information
and National Natural Science Foundation of China Under Grant number 11671158.

\end{document}